\documentclass[12pt]{article}
\pdfoutput=1
\usepackage[breaklinks]{hyperref}
\usepackage{xcolor,caption,subcaption,graphicx}
\usepackage{amsmath}
\usepackage{tikz}
\usepackage{amsthm}
\usepackage{xchicago}
\makeatother
\newcommand {\citeAY} [1] {\citeNP {#1}}%
\newcommand {\citeAPY}[1] {\citeN  {#1}}%
\renewcommand {\showoriginalref}[1]{}
\renewcommand {\showCODEN}[1]{}
\renewcommand {\showISSN}[1]{}
\renewcommand {\showMR}[3]{}

\newcommand\eq[1] {(\ref{#1})}

\newcommand{\bfm}[1]{\mbox{\boldmath ${#1}$}}

\newcommand{\beqa}{\begin{eqnarray}}
\newcommand{\eeqa}[1]{\label{#1}\end{eqnarray}}
\newcommand{\beq}{\begin{equation}}
\newcommand{\eeq}[1]{\label{#1}\end{equation}}
\newcommand{\bbm}{\begin{bmatrix}}
\newcommand{\ebm}{\end{bmatrix}}
\newcommand{\bpm}{\begin{pmatrix}}
\newcommand{\epm}{\end{pmatrix}}

\newcommand{\Divop}{\mathrm{div \,}}
\newcommand{\Curlop}{\mathrm{curl\,}}

\newcommand{\Real}{\mathop{\rm Re}\nolimits}
\newcommand{\Imag}{\mathop{\rm Im}\nolimits}

\newtheorem{Thm}{Theorem}

\newtheorem{Def}[Thm]{Definition}
\newtheorem{Pro}[Thm]{Proposition}
\newtheorem{Cor}[Thm]{Corollary}
\newtheorem{Rem}[Thm]{Remark}

\newtheorem{Asm}[Thm]{Assumption}

\newcommand{\BGG}{\bfm\Gamma}
\newcommand{\BGL}{\bfm\Lambda}
\newcommand{\BGP}{\bfm\Pi}

\newcommand{\CB}{{\cal B}}

\newcommand{\CD}{{\cal D}}
\newcommand{\CE}{{\cal E}}

\newcommand{\CH}{{\cal H}}

\newcommand{\CJ}{{\cal J}}
\newcommand{\CK}{{\cal K}}
\newcommand{\CL}{{\cal L}}
\newcommand{\CM}{{\cal M}}
\newcommand{\CN}{{\cal N}}

\newcommand{\CP}{{\cal P}}

\newcommand{\CU}{{\cal U}}
\newcommand{\CV}{{\cal V}}
\newcommand{\CW}{{\cal W}}

\newcommand{\BCQ}{{\bfm{\cal Q}}}

\def\Be{{\bf e}}

\def\Bj{{\bf j}}

\def\Bn{{\bf n}}

\def\Bu{{\bf u}}
\def\Bv{{\bf v}}

\def\BA{{\bf A}}
\def\BB{{\bf B}}
\def\BC{{\bf C}}

\def\BE{{\bf E}}
\def\BF{{\bf F}}
\def\BG{{\bf G}}

\def\BJ{{\bf J}}
\def\BK{{\bf K}}
\def\BL{{\bf L}}
\def\BM{{\bf M}}

\def\BP{{\bf P}}
\def\BQ{{\bf Q}}
\def\BR{{\bf R}}

\def\BT{{\bf T}}
\def\BU{{\bf U}}

\def\BZ{{\bf Z}}
\def\BLE{{\bf L_*}}
\def\BYE{{\bf Y_*}}

\def \M  {{\cal M}}
\def \ba {\begin{array}}
\def \ea {\end{array}}

\def \refe #1.{(\ref{#1})}
\def \reff #1.{figure~\ref{#1}}
\def \refs #1.{section~\ref{#1}}
\def \refss #1.{subsection~\ref{#1}}
\def \refD #1.{Definition~\ref{#1}}
\def \refT #1.{Theorem~\ref{#1}}
\def \refL #1.{Lemma~\ref{#1}}
\def \refC #1.{Corollary~\ref{#1}}
\def \refP #1.{Proposition~\ref{#1}}
\def \refR #1.{Remark~\ref{#1}}
\def \refE #1.{Example~\ref{#1}}
\def \refN #1.{Notation~\ref{#1}}
\newcommand{\bbC}{{\mathbb{C}}}

\usepackage{graphics}
\usepackage{amssymb}
\begin{document}

\vspace{-1in}
\title{A rigorous approach to the field recursion method for two-component composites with isotropic phases}%
\label{chap:frm}
\author{Maxence Cassier$^{*}$, Aaron Welters$^{**}$ and Graeme W. Milton$^{*}$}
\date{$^{*}$\small{Department of Mathematics, University of Utah, Salt Lake City, UT 84112, USA} \\
$^{**}$\small{Department of Mathematical Sciences, Florida Institute of Technology, Melbourne, FL 32901, USA}\\
\small{Emails: cassier@math.utah.edu, awelters@fit.edu, milton@math.utah.edu}}
\maketitle

\date{}
\maketitle
\vspace{2ex}
%******************************
\begin{abstract}
%******************************
In this chapter we give a rigorous derivation of the field equation recursion method in the abstract theory of composites to two-component composites with isotropic phases. This method is of great interest since it has proven to be a powerful tool in developing sharp bounds for the effective tensor of a composite material. The reason is that the effective tensor $\BLE$ can be interpreted in the general framework of the abstract theory of composites as the $Z$-operator on a certain orthogonal $Z(2)$ subspace collection. The base case of the recursion starts with an orthogonal $Z(2)$ subspace collection on a Hilbert space $\CH$, the $Z$-problem, and the associated $Y$-problem. We provide some new conditions for the solvability of both the $Z$-problem and the associated $Y$-problem. We also give explicit representations of the associated $Z$-operator and $Y$-operator and study their analytical properties. An iteration method is then developed from a hierarchy of subspace collections and their associated operators which leads to a continued fraction representation of the initial effective tensor $\BLE$.
\end{abstract}
\vspace{3ex}
\noindent
\begin{mbox}
\end{mbox}
\vspace{3ex}
\vskip2mm

%DellAntonio:1986:ATO
%%%%%%%%%%%%%%%%%%%%%%%%%%%%%%%%%%%%%%%%%%%%%%%%%%%%%%%%%%%%%%%%%%%%%%%%
\section{Introduction}
\setcounter{equation}{0}
%%%%%%%%%%%%%%%%%%%%%%%%%%%%%%%%%%%%%%%%%%%%%%%%%%%%%%
In this chapter we give a rigorous derivation of the field equation recursion method%
\index{field equation recursion method}
(\citeAY{Milton:1985:TCC}; Milton \citeyearNP{Milton:1987:MCEa}, \citeyearNP{Milton:1987:MCEb}, \citeyearNP{Milton:1991:FER}; \citeAY{Clark:1994:MEC}; \citeAY{Clark:1997:CFR} and Chapter 29 of \citeAY{Milton:2002:TOC} in the abstract theory of composites as described in Chapter 2 of \citeAY{Graeme:2016:ETC}).
The field equation recursion method utilizes information obtained from series expansions%
\index{series expansions}
to derive bounds on the effective
response. The series expansion information is encoded in a series of positive semidefinite weight and normalization matrices.%
\index{weight matrices}%
\index{normalization matrices}
The response could be that of a composite material, resistor network, spring network, or as shown in Chapter 3 of \citeAY{Graeme:2016:ETC}, that
of a multicomponent body with inclusions having a size comparable to that of the body. More generally, it is applicable to any problem that can be phrased,
in the abstract theory of composites, as a $Y$-problem%
\index{Yp@$Y$-problem}
or $Z$-problem,%
\index{Zp@$Z$-problem}
involving subspace collections%
\index{subspace collections}
with orthogonal subspaces: see Chapter 2 of \citeAY{Graeme:2016:ETC}. Thus it is applicable to the many equations discussed in Chapter 1 of \citeAY{Graeme:2016:ETC}, where the key identity is used to establish orthogonality of subspaces see Section
1.19 of Chapter 1 in \citeAY{Graeme:2016:ETC}, and Chapter 3. In particular, as follows from Section 3.12
in Chapter 3, it is applicable to the electromagnetism in multiphase bodies, and
thus provides an alternative to bounds based on variational principles%
\index{variational principles}
or on the analyticity, established in Chapter 3 and more rigorously in Chapter 4, of the Dirichlet-to-Neumann map%
\index{Dirichlet-to-Neumann map}
as functions of the frequency and component moduli. As shown in Chapter 4, the bounds can be used in an inverse way  to give information on the interior structure inside the body.

 For composite materials the $n$-th order series expansion terms, in the expansion of the effective moduli and fields in powers of the
contrast between the materials, can in theory be derived from information derived from $n$-point correlation functions.%
\index{correlation functions}
These incorporate, for example in the case of two-component composites, the information contained in the probability  that a polyhedron lands with all vertices in phase $1$ when dropped randomly in the composite. The series expansions were first derived by \citeAPY{Brown:1955:SMP} for two
phase conducting composites and were subsequently extended by others (\citeAY{Herring:1960:ERI}; \citeAY{Prager:1960:DIM};
\citeAY{Beran:1963:SPE}; \citeAY{Beran:1968:SCT}; \citeAY{Beran:1970:MFV};
\citeAY{Fokin:1969:CEE}; \citeAY{Dederichs:1973:VTE}; \citeAY{Hori:1973:STE}; \citeAY{Zeller:1973:ECP};
\citeAY{Gubernatis:1975:MEP}; \citeAY{Kroner:1977:BEE}; \citeAY{Willis:1981:VRM};
\citeAY{Milton:1982:NBE}; \citeAY{Phan-Thien:1982:NBE}; \citeAY{Sen:1989:ECA}; \citeAY{Torquato:1997:EST};
Tartar \citeyearNP{Tartar:1989:MSA}, \citeyearNP{Tartar:1990:MNA}; and \citeAY{Bruno:1991:TEB}). See also the book of
\citeAPY{Torquato:2002:RHM}, Chapters 14 and 15 in \cite{Milton:2002:TOC}, and Chapter 9 of \citeAY{Graeme:2016:ETC}, where series expansions are derived for coupled field problems. In practice these multipoint correlation functions are difficult to obtain: only three-point correlations can be directly obtained from cross-sectional photographs of three-dimensional materials, since in general four points have one point
lying outside the plane containing the other 3 points. A far more practical way of obtaining the series expansion information, when the geometry is known is to use
Fast Fourier Transform methods%
\index{Fast Fourier Transform methods}
(Moulinec and Suquet \citeyearNP{Moulinec:1994:FNM}, \citeyearNP{Moulinec:1998:NMC}; \citeAPY{Eyre:1999:FNS};
\citeAPY{Willot:2014:FBS}; and Chapter 8 of \citeAY{Graeme:2016:ETC}: when the composite is not periodic one can
take a cubic representative volume element of it, and periodically extend it. For the response of bodies, rather than composites, as discussed in Chapter 3 of \citeAY{Graeme:2016:ETC}, the computation of the series expansion is more difficult, as computation of the action of the nonlocal operators $\BGG_1$ and $\BGG_2$ requires one to know the Green's function%
\index{Green's function}
associated  with the body, with appropriate boundary conditions, and is not simply calculated through Fourier transforms.

For two-component composites the bounds calculated from the field expansion recursion method can alternatively be obtained from
variational principles%
\index{bounds!variational principles}
(\citeAY{Milton:1982:CTM}), or via the analytic method%
\index{bounds!analytic method}
(\citeAY{Bergman:1978:DCC}; \citeAY{Milton:1981:BTO};
\citeAY{Bergman:1993:HSF}). The rapid convergence of the bounds to the actual effective conductivity as successively more series
expansion coefficients are incorporated in the bounds has been shown in numerical studies (\citeAY{McPhedran:1981:BET}; \citeAY{Milton:1982:CTM}).
In a wider context these bounds mostly correspond to bounds on Stieltjes functions,%
\index{Stieltjes functions}
and are related
to Pad\'e approximants%
\index{Pad\'e approximants}
(\citeAY{Gragg:1968:TEB}; \citeAY{Baker:1969:BEB}; \citeAY{Field:1976:SSP}; \citeAY{Baker:1981:PAB}).
The real power of the method lies in its application to multiphase composites, where it allows one to obtain bounds on complex effective moduli that have not been obtained by any other method (see \citeAY{Milton:1987:MCEb} and section 29.6 of \citeAY{Milton:2002:TOC}). In fact many of the bounds are sharp in the context of the abstract theory of composites: there is a subspace collection with an associated response function that saturates the bounds. For real composite materials the differential constraints on the fields typically impose further restrictions on the subspace collections which make many bounds nonoptimal in this context. For example,
for two-dimensional isotropic composites of two isotropic conducting phases, the differential constraints imply that the normalization matrices%
\index{normalization matrices}
are simply identity matrices (and using this fact the bounds can be improved) while for other problems (even three dimensional conductivity, with two isotropic phases) the complete set of constraints on the normalization matrices and weight matrices is unknown.

The differential constraints are incorporated in the translation method%
\index{bounds!translation}
for obtaining bounds (\citeAY{Tartar:1979:ECH}, particularly theorem 8; \citeAY{Murat:1985:CVH}; \citeAY{Tartar:1985:EFC}; Lurie and Cherkaev \citeyearNP{Lurie:1982:AEC}, \citeyearNP{Lurie:1984:EEC} -
see also the books of the books of
\citeAY{Allaire:2002:SOH}; \citeAY{Cherkaev:2000:VMS}; \citeAY{Milton:2002:TOC}; \citeAY{Tartar:2009:GTH}; and \citeAY{Torquato:2002:RHM} and references therein), which however usually has the disadvantage that there are so many ``translations'' to choose from that it is hard to know which ones will produce the best bounds. The field equation recursion method has the advantage of being applicable to any problem that can be formulated in the abstract theory of composites, and the advantage that it produces bounds in a systematic way.

We have in mind effective tensors $\BLE=\BLE(\BL_1,\ldots, \BL_n)$ with tensors $\BL_j$ as variables, but for simplicity we will just consider in this chapter the case $n=2$ in which the $2$-variables are just scalars. Then in the abstract setting an effective tensor is just the $Z$-operator associated to a $Z$-problem on a $Z(n)$ subspace collection (see Chapter 29 in \citeAY{Milton:2002:TOC} and Chapters 2 and 7 of \citeAY{Graeme:2016:ETC}) and provides a solution to the $Z$-problem. The field recursion method is a way to systematically develop, via fractional linear transformations%
\index{fractional linear transformation}
at each level of the recursion, the continued fraction expansions%
\index{continued fraction expansion}
of the effective tensor $\BLE$, i.e., the initial $Z$-operator in the recursion. This method is powerful for developing sharp bounds on the effective tensor analogous to the sharp bounds of Stieltjes functions via its continued fraction representation using Pad\'e approximants%
\index{Pad\'e approximants}
(\citeAY{Gragg:1968:TEB}; \citeAY{Baker:1969:BEB}; \citeAY{Field:1976:SSP}; \citeAY{Baker:1981:PAB}). We caution, however, that the approach using
fractional linear transformations is similar in many ways to the Schur (\citeyearNP{Schur:1917a:UPI},\citeyearNP{Schur:1917b:UPI}) algorithm
(see also section 1.1 in \citeAY{Agler:2002:PIH})
for successively reducing in the number of interpolation points in the Nevanlinna--Pick interpolation problem.%
\index{Nevanlinna--Pick interpolation}
This algorithm has the disadvantage that computational or experimental errors can lead to problems even for a modest number of iterations
(Section 1.1 in \citeAY{Agler:2002:PIH}; \citeAY{Foias:2002:NPI}).

%What we do in this chapter.
The rest of this chapter is organized as follows. In section \ref{QSecFormProb} we introduce, in the Hilbert space setting for the abstract theory of composites, the effective tensor $\BLE(L_1,L_2)$ and its properties as a function of the composite phases $L_1$ and $L_2$. This tensor is defined as the $Z$-operator associated to an orthogonal $Z(2)$ subspace collection in a Hilbert space $\CH$. Next, in section \ref{QSubSecBaseCase} we generate the $\BYE(L_1,L_2)$ tensor on the associated orthogonal $Y(2)$ subspace collection and develop its functional properties as well as prove the important linear fractional transformations that connect $\BLE$ and $\BYE$ together. We then introduce a new effective tensor $\BLE^{(1)}(L_1,L_2)$ on a new orthogonal $Z(2)$ subspace collection in a Hilbert space $\CH^{(1)}$ which is a subspace of the initial Hilbert space $\CH$. The base case of the recursion is then concluded by showing that these two tensors $\BYE$ and $\BLE^{(1)}$ are congruent (i.e., there exists an invertible operator $\BK$ such that $\BYE=\BK\BLE^{(1)}\BK^{\dag}$).
Finally, in section \ref{QSubsecInductionStep} we prove via induction (under additional assumptions) that we can repeat this process. In other words, we
can define an effective tensor $\BLE^{(i)}(L_1,L_2)$ on some $Z(2)$ subspace collection in a Hilbert space $\CH^{(i)}$ which is a subspace of $\CH^{(i-1)}$. This method provides a continued fraction representation of the effective tensor $\BLE$ in terms of the $\BLE^{(i)}$ and $\BYE^{(i)}$. Thus, if one can provide bounds at each level of this representation then one can obtain very tight bounds on $\BLE$, which is one of the most important applications of this field recursion method (\citeAY{Milton:1987:MCEb}; see also
Chapter 9 of \citeAY{Graeme:2016:ETC}).

A major achievement in this chapter is the use of the theory of Fredholm operators%
\index{Fredholm operator}
which allows us to significantly reduce, at each level of the induction, the necessary assumptions required to define the operators $\BLE^{(i)}$ and $\BYE^{(i)}$. And although our approach here is abstract, which simplifies mathematically many aspects of the field equation recursion method, we want to keep in mind that its development is motivated by its applications in the theory of composites. As such, we use throughout this chapter the quintessential example of the effective conductivity tensor for composite materials in quasistatic electromagnetism [see, for instance, \citeAY{Milton:1985:TCC}; \citeAY{Milton:1991:FER}; \citeAY{Clark:1994:MEC}; \citeAY{Clark:1995:OBC};
  \citeAY{Clark:1997:CFR}; \citeAY{Milton:2002:TOC}] in order to illustrate the theory.

This chapter assumes the reader is familiar with Chapters 1 and 2 of \citeAY{Graeme:2016:ETC}. Chapter 9 of \citeAY{Graeme:2016:ETC} is closely related to the subject matter of the chapter,
although need not be read beforehand. Like Chapter 4 of \citeAY{Graeme:2016:ETC} it is written in a rigorous mathematical style and care has been taken to explain most technical
definitions to ensure it is accessible to non-mathematicians.  

Before we proceed, let us introduce some notation and definitions:
\begin{itemize}
	\item We denote by $\CB(\CH)$ the Banach space of all bounded linear operators on a Hilbert space $\CH$ and endow this space with the operator norm.
	\item If $\BT \in \CB(\CH)$, we denote respectively by $\operatorname{Ran}\BT$, $\operatorname{Ker}\BT$ and $\BT^{\dag}$ its range, kernel (i.e., nullspace) and adjoint operator.
	\item If $\BT \in \CB(\CH)$, $\operatorname{Re}\BT$ and $\operatorname{Im}\BT$ stand for the bounded operators
	$$
	\operatorname{Re}\BT=\frac{\BT+\BT^{\dag} }{2} \, \mbox{ and } \, \operatorname{Im}\BT=\frac{\BT-\BT^{\dag} }{2i}.
	$$
\item An operator $\BT \in  \CB(\CH)$ is said to be positive definite (i.e., $\BT>0$) if
\begin{equation}\label{Qeq.posdef}
  \forall \Bu \in \CH\setminus \{{\bf 0} \}, \ (\Bu, \BT \Bu)>0,
\end{equation}
and negative definite if $-\BT$ is positive definite. If the inequality ``$>$'' in (\ref{Qeq.posdef}) is instead replaced by ``$\geq$'', then one says that $\BT$ is positive semidefinite.
\item An operator $\BT \in  \CB(\CH)$ is said to be coercive%
\index{coercivity property}
(or uniformly positive)%
\index{operators!uniformly positive}
if
$$
\exists \alpha>0 \mid   \ \forall \Bu \in \CH, \  (\Bu, \BT \Bu)\geq \alpha \, (\Bu,\Bu),
$$
and in that case we write $\BT \geq \alpha \mathbb{I}$. An operator is said uniformly negative%
\index{operators!uniformly negative}
if $-\BT$ is coercive.
	\item We denote by $\CE^{\perp}$ the orthogonal complement of a subspace $\CE$ of $\CH$,
	\item $\operatorname{cl}S$ stands for the topological closure of a set $S$ in a Hilbert space $\CH$ in the metric topology generated by the inner product norm.
\end{itemize}

\begin{Def}(multivariate analyticity)\label{QDef.Analyticity}
Let $U$ be an open set of $\bbC^{n}$ and $F$ a complex Banach space.  A function $f:U \to \BF$ is said to be analytic if it is differentiable on $\CU$. If we denote by $\BZ=(z_1,\cdots,z_n)$ the points of $U$, then by Hartogs' theorem%
\index{Hartogs' theorem}
(see Theorem 36.8 in chapter VIII, section 36, p.\ 271 of \citeNP{Mujica:1986:CAB}), $f$ is analytic in $U$ if and only if it is a differentiable function in each variable $z_i$ separately.
\end{Def}

\begin{Def}
Let $\BT\in \CB(\CH)$, we say that the operator $\BT$ is Fredholm%
\index{Fredholm operator}
if $\operatorname{Ran}\BT$ is a closed subspace of $\CH$ and
if $\operatorname{Ker}\BT$ and $\operatorname{Ker}\BT^{\dag}$ are finite dimensional spaces. In that case, the index%
\index{Fredholm operator!index}
 $\operatorname{ind}{\BT}$ of $\BT$ is defined by $$\operatorname{ind}{\BT}= \dim \operatorname{Ker}{\BT}-\dim \operatorname{Ker}{\BT^{\dag}}.$$
\end{Def}
\begin{Rem}
Using the well-known identity $\operatorname{cl}\operatorname{Ran}{\BT}=(\operatorname{Ker}{\BT^{\dag}})^{\perp}$, it follows that a Fredholm operator of index $0$%
\index{Fredholm operator!index $0$}
is invertible if and only if $\operatorname{Ker}\BT=0$.
\end{Rem}
We will use the following propositions:
\begin{Pro}\label{QPropFredholmSumEquiv}
An operator $\BT\in \CB(\CH)$ is a Fredholm operator of index $0$ if and only if can be written as $\BT=\BA+\BK$ with $\BA$ an invertible operator and $\BK$ a finite-rank operator,%
\index{operators!finite rank}
i.e., $\dim \operatorname{Ran}\BK < \infty$.
\end{Pro}
The proof of this result can be found in chapter XV, section 15.2, p.\ 350, Corollary 2.4 in \citeAY{Gohberg:2003:BCO}.
\begin{Pro}\label{QPropFredholmBlockEquiv}
Suppose $\CH=\CH_1\oplus\CH_2$ is an orthogonal decomposition of a Hilbert space $\CH$ with $\dim \CH_1<\infty$. Let $\BT\in \CB(\CH)$ be written, with respect to this decomposition, as a $2\times 2$ block operator matrix form
\begin{align}
\BT=\begin{bmatrix}
\BT_{11} & \BT_{12}\\
\BT_{21} & \BT_{22}
\end{bmatrix}.
\end{align}
If $\BT$ is a Fredholm operator of index $n$ then $\BT_{22}\in\CB(\CH_2)$ is a Fredholm operator of index $n$. In particular, if $\BT$ is invertible then $\BT_{22}$ has index $0$.
\end{Pro}
\begin{proof}
Suppose that $\BT$ is a Fredholm operator of index $n$ then since $\dim \CH_1<\infty$, the operator
\begin{align}
\begin{bmatrix}
\BT_{11} & \BT_{12}\\
\BT_{21} & 0
\end{bmatrix}
\end{align}
is a finite-rank operator in $\CB(\CH)$ which implies by the Proposition \ref{QPropFredholmSumEquiv} that
\begin{align}
\begin{bmatrix}
0 & 0\\
0 & \BT_{22}
\end{bmatrix}
= \BT-
\begin{bmatrix}
\BT_{11} & \BT_{12}\\
\BT_{21} & 0
\end{bmatrix}
\end{align}
is a Fredholm operator in $\CB(\CH)$ of index $n$. Thus,
\begin{align}
\CH_1\oplus\operatorname{Ker}\BT_{22},\;\;\CH_1\oplus\operatorname{Ker}\BT_{22}^{\dag},
\end{align}
are finite-dimensional, $\operatorname{Ran}\BT_{22}$ is closed, and $$n=\dim(\CH_1\oplus\operatorname{Ker}\BT_{22})-\dim(\CH_1\oplus\operatorname{Ker}\BT_{22}^{\dag})=\dim(\operatorname{Ker}\BT_{22})-\dim(\operatorname{Ker}\BT_{22}^{\dag}).$$ Therefore, $\BT_{22}$ is a Fredholm operator of index $n$. In particular, if $\BT$ is invertible then it is a Fredholm operator of index $0$ and so is $\BT_{22}$. This completes the proof.
\end{proof}
%%%%%%%%%%%%%%%%%%%%%%%%%%%%%%%%%%%%%%%%%%%%%%%%%%%%%%%%%%%%%%%%%%%%%%%%%%%%%%%%%%%%%%%
\section{Formulation of the problem for two-component composites}\label{QSecFormProb}
\setcounter{equation}{0}
%%%%%%%%%%%%%%%%%%%%%%%%%%%%%%%%%%%%%%%%%%%%%%%%%%%%%%
%%%%%%%%%%%%%%%%%%%%%%%%%%%%%%%%%%%%%%%%%%%%%%%%%%%%%%%%%%%%%%%%%%%%%%%%%%%%%%%%%%%%%%%
The setting from Chapter 2 of \citeAY{Graeme:2016:ETC} is a Hilbert space $\CH$ which has an inner product $(\BP_1,\BP_2)$ defined for all $\BP_1,\BP_2\in \CH$ having the usual properties that
\begin{align}\label{QDefInnerProduct}
(\BP_1,\BP_2)=\overline{(\BP_2,\BP_1)},\;\;(\BP_1,\BP_1)>0\;\;\textnormal{for all } \BP_1\not=0,
\end{align}
with the convention that it is linear in the second component and antilinear in the first.
 This Hilbert space is assumed to have the decomposition
\begin{align}\label{QZ2SubspColl}
\CH = \CU\oplus\CE\oplus\CJ=\CP_1\oplus\CP_2,
\end{align}
where the subspaces $\CU$, $\CE$, and $\CJ$ are mutually orthogonal with respect to this inner product as are the subspaces $\CP_1$, $\CP_2$. Moreover, we assume $\CU$ is a finite-dimensional subspace and that $\CP_1$ and $\CP_2$ are nonzero subspaces. Under these assumptions, the decomposition in (\ref{QZ2SubspColl}) is called an orthogonal $Z(2)$ subspace collection%
\index{subspace collections!orthogonal}
(see Chapter 7 of \citeAY{Graeme:2016:ETC}).

We denote by $\BGG_0$, $\BGG_1$ and $\BGG_2$%
\index{projection operators!gamma@$\BGG_0$, $\BGG_1$, and $\BGG_2$}
the orthogonal projections on $\CU$, $\CE$, and $\CJ$, respectively, and denote by $\BGL_1$ and $\BGL_2$ the orthogonal projections onto $\CP_1$ and $\CP_2$, respectively.
We want to emphasize that we don't assume that the operators $\BGL_a$ and $\BGG_i$ commute. In this abstract setting, the Hilbert space $\CH$ can be infinite-dimensional and as such it may not be clear that the subspaces introduced above must all be closed subspaces so that their corresponding orthogonal projections on each subspace exist. These facts though follow immediately from the following classical proposition which we prove here so that this chapter is self-contained.
\begin{Pro}
Let $\CH$ be a Hilbert space with inner product (\ref{QDefInnerProduct}). If $\CM_1$ and $\CM_2$ are subspaces of $\CH$ which are mutually orthogonal and $\CH=\CM_1\oplus\CM_2$ then $\CM_1$ and $\CM_2$ are closed sets in the norm topology on $\CH$ and, for $i,j=1,2$ with $i\not=j$,  the projection $\BQ_i:\CH\rightarrow \CH$ onto $\CM_i$ along $\CM_j$, which is uniquely defined by $\BQ_i\BP=\BP_i$ if $\BP=\BP_1+\BP_2$ for some vectors $\BP_1\in \CM_1$ and $\BP_2\in \CM_2$, is an orthogonal projection.
\end{Pro}
\begin{proof}
Suppose $\CM_1$ and $\CM_2$ are subspaces of $\CH$ which are mutually orthogonal and $\CH=\CM_1\oplus\CM_2$. Then every $\BP\in\CH$ can be written as $\BP=\BP_1+\BP_2$ for some unique $\BP_1\in \CM_1$, $\BP_2\in\CM_2$ and there exists unique linear operators $\BQ_i:\CH\rightarrow \CH$ defined by $\BQ_i\BP=\BP_i$, $i=1,2$. In particular, this implies that they are projections, i.e., $\BQ_i^2=\BQ_i$, $i=1,2$, with the property that they sum to the identity operator, i.e., $\BQ_1+\BQ_2=I_{\CH}$. It now follows from the fact that the spaces $\CM_1,\CM_2$ are mutually orthogonal that for any $\BP\in\CH$ with $\BP=\BP_1+\BP_2$ for some $\BP_1\in \CM_1$, $\BP_2\in\CM_2$ we have  $||\BQ_i\BP||=||\BP_i||\leq||\BP_1||+||\BP_2||=||\BP||$ and $||\BQ_i\BP_i||=||\BP_i||$ for $i=1,2$. Thus, since $\CH=\CM_1\oplus\CM_1$, this implies that in the operator norm $||\BQ_i||=1$ or $\BQ_i=0$ and hence $\BQ_i\in\CB(\CH)$ for $i=1,2$. Thus, since $\CB(\CH)$ is the space of all continuous linear operator from $\CH$ into $\CH$ in the norm topology on $\CH$ then this implies $\BQ_i^{-1}(\{{\bf 0}\})=\CM_i$ is a closed set. Finally, since the spaces $\CM_1$ and $\CM_2$ are mutually orthogonal then for any $\BP_i, \BR_i\in\CM_i$, $i=1,2$ we have $\left(\BQ_i(\BP_1+\BP_2),\BR_1+\BR_2\right)=\left(\BP_i,\BR_i\right)=\left(\BP_1+\BP_2,\BQ_i(\BR_1+\BR_2)\right)$
which implies that $\BQ_i^{\dag}=\BQ_i$, i.e., the projection operator $\BQ_i$ is self-adjoint, that is, it is an orthogonal projection.  This completes the proof.
\end{proof}

We now define the bounded linear operator-valued function on $\CH$ by
\begin{align}\label{QLFunction}
\BL=\BL(L_1,L_2)=L_1\BGL_1+L_2\BGL_2.
\end{align}
It is a function of the complex variables $L_1, L_2$, and $\BL:\bbC^2\rightarrow \CB(\CH)$ is an analytic $\CB(\CH)$-valued function in the sense of the Definition \ref{QDef.Analyticity}.

\begin{Rem}\label{Qrem.conduct1}
Defining this functional framework for the conductivity equation%
\index{conductivity equations}
of a two-component periodic composite defined by a unit cell $\CD$, which is an open simply connected set with a Lipschitz continuous boundary, and composed of isotropic materials characterized by their complex conductivity $\sigma_1 $ and $\sigma_2$ leads to:
\begin{itemize}
\item $\CH=\BL^2(\CD)=[L^2(\CD)]^d$ the space of square-integrable, $d$-dimensional, vector-field-valued functions on the unit cells
$\CD$ endowed with the inner product
$$
(\BE,\BE^{'})=\frac{1}{|\CD|}\int_{\CD}\BE^{\top}\overline{\BE^{'}} \, \mathrm{d} {\bf x},
 \ \forall \BE, \BE^{'} \in  \M_{d,1}(\bbC),$$
where $|\CD|$ denotes the volume of $\CD$,
\item $\CU$ is the $d$-dimensional space defined by $$\CU=\left\{\BU\in\CH\mid \BU\equiv C, \mbox{ for some  } C\in \M_{d,1}(\bbC) \right\},$$
\item $\CE=\left\{ \BE \in  \CH \mid \Curlop \BE=0  \mbox{ in } \CD \mbox{ and }  \int_{\CD}\BE \, \mathrm{d} {\bf x}=0  \mbox{ (zero average condition)}\right\},$
\item $\CJ=\left\{ \BJ \in  \CH \mid \Divop \BJ =0 \mbox{ in } \CD \mbox{ and }  \int_{\CD}\BJ \, \mathrm{d} {\bf x}=0\right\}.$
\end{itemize}
For a proof of the orthogonal decomposition $\BL^{2}(\CD)= \CU\oplus \CE\oplus \CJ$, we refer to
Chapter 12 of \citeAPY{Milton:2002:TOC}.
Although the definition (\ref{QLFunction}) of the operator $\BL$ may seem abstract, one sees it is natural in the theory of composites (see \citeAY{Milton:2002:TOC}) where $\BL$ is the conductivity tensor
\begin{align}
\BL(\sigma_1, \sigma_2) = \sigma_1\chi_1\mathbb{I}+\sigma_2\chi_2\mathbb{I} \, .
\end{align}
In this context, $\chi_i$ is the indicator function%
\index{indicator function}
of the domain occupied by phase $i$ and the scalars $L_1$ and $L_2$ in (\ref{QLFunction}) are respectively the complex conductivities $\sigma_1$ and $\sigma_2$. Thus, $\BGL_1=\chi_1\mathbb{I}$ and $\BGL_2=\chi_2\mathbb{I}$ represent respectively the projections onto the spaces $\CP_1$ and $\CP_2$ of $~\BL^2(\CD)$ fields which are nonzero only inside phase $1$ or $2$. Finally, it should be emphasized that the restriction of $~\BL$ on the subspace $\CU\oplus \CE$ represents the constitutive law of the composite which links an electric field in $\CU\oplus \CE$ to a current density in $\CU\oplus \CJ$.
\end{Rem}

Following Chapter 7 of \citeAY{Graeme:2016:ETC}, associated to this $Z(2)$ subspace collection (\ref{QZ2SubspColl}) is a linear operator-valued function $\BLE(L_1,L_2)$ (i.e., the associated $Z$-operator) acting on the subspace $\CU$. To obtain this function we begin by solving the following problem: for a given $\Be\in\CU$, find a unique triplet of vectors $\Bj\in\CU$, $\BE\in\CE$, and $\BJ\in\CJ$ that satisfy
\begin{align}\label{QZProblem}
\Bj+\BJ = \BL(\Be+\BE),
\end{align}
also known as the $Z$-problem%
\index{Zp@$Z$-problem}
(see Chapter 7 of \citeAY{Graeme:2016:ETC}).
The associated operator $\BLE$, by definition, governs the linear relation
\begin{align}\label{QZOperator}
\Bj=\BLE\Be.
\end{align}
To obtain a representation of $\BLE$ using $\BL$, we follow the approach of Chapter 12 of \citeAPY{Milton:2002:TOC}  which consists of expressing
$\BL$ as a $3\times 3$ block operator matrix
\begin{align}
\BL=\begin{bmatrix}
\BL_{00} & \BL_{01} & \BL_{02} \\
\BL_{10} & \BL_{11} & \BL_{12}  \\
\BL_{20} & \BL_{21} & \BL_{22}
\end{bmatrix}
\end{align}
with respect to the decomposition $\CH = \CU\oplus\CE\oplus\CJ$.
Solving the $Z$-problem (\ref{QZProblem}) for a given $\Be\in \CU$ is then equivalent to proving that the system
\begin{equation}\label{Qeq.systLstar}
\left\{ \begin{array}{ll}
\BL_{00} \Be+  \BL_{01} \BE=\Bj,  \\[5pt]
\BL_{10} \Be+  \BL_{11} \BE=0 ,\\[5pt]
\BL_{20} \Be+ \BL_{21} \BE= \BJ ,
\end{array} \right.
\end{equation}
admits a unique solution $(\Bj,\BE,\BJ) \in \CU \times \CE \times \CJ $. Under the assumption that the block operator $\BL_{11}$ (which represent the restriction of the operator $\BGG_1\BL \BGG_1$ on the subspace $\CE$) is invertible, the solution of the system (\ref{Qeq.systLstar}) is unique and given by
\begin{equation}\label{Qeq.systLstar1}
\left\{ \begin{array}{ll}
\Bj= (\BL_{00}- \BL_{01}  \BL_{11}^{-1} \BL_{10} )\, \Be , \\[5pt]
 \BE= - \BL_{11}^{-1}\,\BL_{10} \,\Be ,\\[5pt]
 \BJ =\BL_{20}\, \Be+ \BL_{21} \,\BE,
\end{array} \right.
\end{equation}
which defines the $Z$-operator $\BLE$ via a Schur complement%
\index{Schur complement}%
\index{effective!operator!L@$\BL_*$ formula}
by:
\begin{equation}\label{QZOperatorFormulaschur}
\BLE=(\BL_{00}- \BL_{01}  \BL_{11}^{-1} \BL_{10} )=  \BGG_0\BL \BGG_0- \BGG_0\BL \BGG_1(\BGG_1\BL \BGG_1)^{-1} \BGG_1\BL \BGG_0,
\end{equation}
where the second equality in the last relation has to be restricted to the subspace $\CU$ of $\CH$.

Another representation formula of the operator $\BLE$:%
\index{effective!operator!L@$\BL_*$ formula}
\begin{align}\label{QZOperatorFormula}
\BLE=\BGG_0[(\BGG_0+\BGG_2)\BL^{-1}(\BGG_0+\BGG_2)]^{-1}\BGG_0,
\end{align}
is proved in Section 12.8 of \citeAPY{Milton:2002:TOC}, equation (12.59) (see also Section 7.4, equation (7.68) of \citeAY{Graeme:2016:ETC}) under the assumption that the inverse of $(\BGG_0+\BGG_2)\BL^{-1}(\BGG_0+\BGG_2)$ exists on the subspace $\CU\oplus\CJ$ (which requires in particular that $\BL$ is invertible on $\CH$).
\begin{Rem}\label{QRem.pos}
For instance, if $\Imag\BL\ge\alpha\mathbb{I}$ or $\Real\BL\ge\alpha\mathbb{I}$ for some $\alpha>0$ then  the inverse of $\BGG_1\BL \BGG_1$ on $\CE$ and the inverse of $(\BGG_0+\BGG_2)\BL^{-1}(\BGG_0+\BGG_2)$ on $\CU\oplus\CJ$ exists and  thus the $Z-$ problem (\ref{QZProblem}) is well-defined and both representation formulas (\ref{QZOperatorFormulaschur}) and  (\ref{QZOperatorFormula} ) of $\BLE$ hold.
\end{Rem}
Our next proposition shows that these two representations of $\BLE$ are both well-defined together or neither is when $\BL^{-1}$ exists.
\begin{Pro}\label{Qprop.equivrepresLE}
If $\BL(L_1,L_2)$ is invertible then the operator $\BA:=\BGG_1\BL(L_1,L_2)\BGG_1:\CE\rightarrow\CE$ is a Fredholm operator of index $n$ if and only if the operator $\BB:=(\BGG_0+\BGG_2)\BL(L_1,L_2)^{-1}(\BGG_0+\BGG_2):\CU\oplus\CJ\rightarrow\CU\oplus\CJ$ is a Fredholm operator of index $n$. Furthermore, the operator $\BA$ is invertible if and only if $\BB$ is invertible, and in that case we have
\begin{align}
\BA^{-1}&=\BGG_1\BL(L_1,L_2)^{-1}\BGG_1-\BGG_1\BL(L_1,L_2)^{-1}(\BGG_0+\BGG_2)\BB^{-1}(\BGG_0+\BGG_2)\BL(L_1,L_2)^{-1}\BGG_1,\\
\BB^{-1}&=(\BGG_0+\BGG_2)\BL(L_1,L_2)(\BGG_0+\BGG_2)-(\BGG_0+\BGG_2)\BL(L_1,L_2)\BGG_1\BA^{-1}\BGG_1\BL(L_1,L_2)(\BGG_0+\BGG_2).
\end{align}
\end{Pro}
\begin{proof}
Suppose $\BL(L_1,L_2)$ is invertible. Then with respect to the orthogonal decomposition of the Hilbert space $\CH=\CE\oplus(\CU\oplus\CJ)$, with the corresponding orthogonal projections $\BGG_1$ and $\BGG_0+\BGG_2$, the operator $\BL(L_1,L_2)$ and its inverse $\BL(L_1,L_2)^{-1}$ can be written as the $2\times 2$ block operator matrices
\begin{align}
\BL(L_1,L_2)=\begin{bmatrix}
\BA & \BA_{12}\\
\BA_{21} & \BA_{22}
\end{bmatrix},\;\;
\BL(L_1,L_2)^{-1}=\begin{bmatrix}
\BB_{11} & \BB_{12}\\
\BB_{21} & \BB
\end{bmatrix}.
\end{align}
This means that $\BA$ and $\BB$ are matricially coupled operators%
\index{operators!matricially coupled}
and therefore the proof of this proposition now follows immediately from Corollary 4.3, pp.\ 46--47, section III.4 of \citeAY{Gohberg:1990:CL1}.
\end{proof}
\begin{Rem}\label{QRem.inv}
The invertibility condition of $\BL(L_1,L_2)$ is straightforward to check.
Indeed, we have
\begin{gather*}
\BL(L_1,L_2) \mbox{ is invertible if and only if } L_1\neq 0 \mbox{ and } L_2\neq 0.
\end{gather*}
Moreover, the operator $\BL(L_1,L_2)^{-1}$  is given by the following formula:
$$
\BL(L_1,L_2)^{-1}=\frac{1}{L_1} \Lambda_1+\frac{1}{L_2} \Lambda_2.
$$
Nevertheless, we emphasize that the invertibility of $~\BL$ on $\CH$ does not imply in general the invertibility of $~\BGG_1 \BL \BGG_1$ on the subspace $\CE$.
\end{Rem}

It follows from the formulas (\ref{QZOperatorFormulaschur}) or  (\ref{QZOperatorFormula}) that the analytic $\CB(\CU)$-valued function $\BLE(L_1,L_2)$ (see Definition \ref{QDef.Analyticity}) satisfies the following homogeneity, normalization, and Herglotz properties:

{\bf The homogeneity property:}%
\index{homogeneity property}
\begin{align}\label{QHomProperty}
\BLE(L_1,L_2)=\frac{1}{c}\BLE(cL_1,cL_2),
\end{align}
for all choices of constants $c\not=0$.

{\bf The normalization property:}%
\index{normalization property}
\begin{align}\label{QNormProperty}
\BLE(1,1)=\mathbb{I}_{\CH}.
\end{align}

{\bf The Herglotz property:}%
\index{Herglotz property}
\begin{align}\label{QHerglotzProperty}
\Imag\BLE(L_1,L_2)>0,\;\;\textnormal{when } \Imag(L_1)>0\textnormal{ and }\Imag(L_2)>0.
\end{align}
The Herglotz property of $\BLE$ may not be obvious and so we provide two proofs. The first one depends on the representation formula (\ref{QZOperatorFormula}) of $\BLE$. The second one is independent of the representation formulas  (\ref{QZOperatorFormulaschur}) or  (\ref{QZOperatorFormula}) (which may not exist) and just assumes that $\BLE$ is well-defined by the $Z$-problem  (\ref{QZProblem}). Nevertheless, as it was mentioned in Remark \ref{QRem.pos}, this coercivity assumption%
\index{coercivity property}
has the convenience to be a sufficient condition of the well-posedness of the $Z$-problem and to justify the existence of the both representation formulas (\ref{QZOperatorFormulaschur}) and  (\ref{QZOperatorFormula}) of $\BLE$.
\begin{proof}
By the definition (\ref{QLFunction}) of the operator $\BL$, it is straightforward that the conditions $\Imag(L_1)>0$ and $\Imag(L_2)>0$ imply that $\Imag(\BL)\ge\beta  \mathbb{I}$ with $\beta=\min(\Imag(L_1), \Imag(L_2))>0$. Thus, by the Lax--Milgram theorem%
\index{Lax--Milgram theorem}
one deduces immediately that $\BL$ is invertible on $\CH$ and by virtue of the identities:
\begin{align}\label{QFundImProps}
\Imag(\BT^{-1})= - (\BT^{-1})^{\dag} \Imag(\BT) \BT^{-1},\;\;\Imag(\BM^{\dag}\BA\BM)=\BM^{\dag}\Imag(\BA)\BM,
\end{align}
which hold for every invertible operator $\BT$ of $\CB(H)$ and every $\BA,\BM\in \CB(H)$,
we get immediately that $\Imag(\BGG_0+\BGG_2)\BL^{-1}(\BGG_0+\BGG_2)$ is a uniformly negative operator on the subspace $\CU\oplus\CJ$ and thus (by  the Lax--Milgram theorem) $(\BGG_0+\BGG_2)\BL^{-1}(\BGG_0+\BGG_2)$ is invertible on $\CU\oplus \CJ$ and hence the operator $\BLE$ is well-defined by (\ref{QZOperatorFormula}) if $\Imag(L_1)>0$ and $\Imag(L_2)>0$. Moreover, by using again (\ref{QFundImProps}), one concludes that $\Imag\BLE=\Imag\BGG_0[(\BGG_0+\BGG_2)\BL^{-1}(\BGG_0+\BGG_2)]^{-1}\BGG_0$ is coercive on $\CU$ (in particular this implies that $\BLE$ is invertible).

The second version of this proof is based on the definition (\ref{QZProblem}) and (\ref{QZOperator}) of the $Z$-problem, and follows the treatment given at the end of
Section 2.6 of \citeAY{Graeme:2016:ETC} (see also the end of Section 12.10 in \citeAY{Milton:2002:TOC}).  Assume that the $\BLE(L_1,L_2)$ is well-defined by the $Z$-problem. Let $\Be\in \CU$ then there exists a $\BE\in \CE$, $\BJ\in\CJ$, and $\Bj\in\CU$ such that $\BLE(L_1,L_2)(\Be+\BE)=\Bj+\BJ$. Hence we obtain
\beqa
( (\Be+\BE), \operatorname{Im} \BL (\Be+\BE))&=&\operatorname{Im} \big(\Be+\BE, \BL(\Be+\BE)\big) \nonumber \\
&=& \operatorname{Im} (  \Be+\BE, \Bj+\BJ)\nonumber \\
&=& \operatorname{Im} (\Be, \BLE \Be)\nonumber\\
&=& (\Be, \operatorname{Im} \BLE \Be).
\eeqa{QLposit}
Thus, the positive definiteness of  $\operatorname{Im}\BL(L_1,L_2)$ under this hypothesis $\Imag(L_1)>0$ and $\Imag(L_2)>0$ implies immediately the coercivity of $\operatorname{Im}[\BLE(L_1,L_2)]$ as $\CU$ is finite-dimensional, in particular, it also implies that $\BLE(L_1,L_2)$ is invertible.
\end{proof}
%%%%%%%%%%%%%%%%%%%%%%%%%%%%%%%%%%%%%%%%%%%%%%%%%%%%%%%%%%%%%%%%%%%%%%%%%%%%%%%%%%%%%%%%%%%%%%%%%%
\section{Field equation recursion method for two-component composites}\label{QSecFieldRecursion}
\setcounter{equation}{0}
%%%%%%%%%%%%%%%%%%%%%%%%%%%%%%%%%%%%%%%%%%%%%%%%%%%%%%
%%%%%%%%%%%%%%%%%%%%%%%%%%%%%%%%%%%%%%%%%%%%%%%%%%%%%%%%%%%%%%%%%%%%%%%%%%%%%%%%%%%%%%%%%%%%%%%%%
\subsection{The base case}\label{QSubSecBaseCase}
In the field equation recursion method,%
\index{field equation recursion method}
the first step is to generate, from the orthogonal $Z(2)$ subspace collection (\ref{QZ2SubspColl}), an orthogonal $Y(2)$ subspace collection (see Chapter 9 of \citeAY{Graeme:2016:ETC}) namely, the Hilbert space
\begin{align}\label{QYSubspaceCollection}
\CK =\CE\oplus\CJ=\CV\oplus\CP_1^{(1)}\oplus\CP_2^{(1)},
\end{align}
where $\CK$ is the orthogonal complement of $\CU$ in $\CH$, i.e.,
\begin{align}
\CK = \CH\ominus\CU.
\end{align}
Then one generates another orthogonal $Z(2)$ subspace collection, namely, the Hilbert space
\begin{align}\label{QZ1SubspaceCollection}
\CH^{(1)}=\CU^{(1)}\oplus\CE^{(1)}\oplus\CJ^{(1)}=\CP_1^{(1)}\oplus\CP_2^{(1)},
\end{align}
with a suitable choice of the spaces, satisfying the inclusion relations
\begin{align}\label{QSubspaceInclusionRelations}
\CE^{(1)}\subseteq \CE,\;\;\CJ^{(1)}\subseteq \CJ,\;\;\CP_1^{(1)}\subseteq \CP_1,\;\;\CP_2^{(1)}\subseteq \CP_2.
\end{align}

We have already defined $\CK =\CE\oplus\CJ$, so our goal is to define the spaces $\CP_1^{(1)}$, $\CP_2^{(1)}$, $\CV$, $\CE_1^{(1)}$, $\CE_2^{(1)}$, $\CJ_1^{(1)}$, $\CJ_2^{(1)}$, $\CH^{(1)}$, and $\CU^{(1)}$. The spaces $\CP_1^{(1)}$ and $\CP_2^{(1)}$ are defined as
\begin{align}
\CP_1^{(1)}=\CP_1\cap \CK,\;\;\CP_2^{(1)}=\CP_2\cap \CK.
\end{align}
$\CV$ is the orthogonal complement of $\CH^{(1)}=\CP_1^{(1)}\oplus\CP_2^{(1)}$ in $\CK$, i.e.,
\begin{align}
\CV=\CK\ominus\CH^{(1)},
\end{align}
the spaces $\CE^{(1)}$ and $\CJ^{(1)}$ are defined as
\begin{align}
\CE^{(1)}=\CE\cap \CH^{(1)},\;\;\CJ^{(1)}=\CJ\cap \CH^{(1)},
\end{align}
and finally $\CU^{(1)}$ is defined as the orthogonal complement of $\CE^{(1)}\oplus\CJ^{(1)}$ in the space $\CH^{(1)}$, i.e.,
\begin{align}\label{Qeq.defCU1}
\CU^{(1)}=\CH^{(1)}\ominus[\CE^{(1)}\oplus\CJ^{(1)}].
\end{align}

The following projections $\BGG_1$, $\BGG_2$, $\BGP_1$, $\BGP_2$, $\BGL_1^{(1)}$, and $\BGL_2^{(1)}$ play a key role and are respectively defined as the orthogonal projections onto $\CE$, $\CJ$, $\CV$, $\CH^{(1)}$, $\CP_1^{(1)}$, and $\CP_2^{(1)}$.

\begin{Rem}\label{QRemLOnSubspaceRestr}
By the subspace inclusion (\ref{QSubspaceInclusionRelations}) we have
\begin{gather}
\BL=L_1\BGL_1+L_2\BGL_2 = L_1\BGL_1^{(1)}+L_2\BGL_2^{(1)}\;\;\textnormal{on } \CH^{(1)}=\CP_1^{(1)}\oplus \CP_2^{(1)}\\
\BL=L_2\BGL_2^{(1)}\;\;\textnormal{on }
\CP_2^{(1)},\;\;\BL=L_2\BGL_1^{(1)}\;\;\textnormal{on } \CP_1^{(1)}.
\end{gather}
\end{Rem}

\begin{Rem}\label{Qrem.conduct2}
The interpretation of the functional spaces which appear in the two subspace collections  $Y(2)$ and $Z(2)$ in (\ref{QYSubspaceCollection}) and (\ref{QZ1SubspaceCollection}) in the case of the conductivity equation for a two-component composite are:
\begin{itemize}
\item  $\CP_i^{(1)}=\CP_i\cap \CK$ for $i=1,2$ is the subspace of $~\BL^{2}(\CD)$ of all fields whose support is included in the phase $i$ and have zero average in the unit cell $\CD$,
\item $\CH^{(1)}=\CP_1^{(1)}\oplus \CP_2^{(1)}$ is the subspace of $~\BL^{2}(\CD)$ of all fields which have zero average in each phase,
\item $\CV$ is the subspace of $~\BL^{2}(\CD)$ of all fields which are constant in each phase and have zero average in the unit cell $\CD$,
\item $\CE^{(1)}=\left\{ \BE \in  \CH^{(1)} \mid \Curlop \BE=0  \mbox{ in } \CD \right\},$ $\CJ^{(1)}=\left\{ \BJ \in  \CH^{(1)} \mid \Divop \BJ =0 \mbox{ in } \CD  \right\}.$
\end{itemize}
The definition (\ref{Qeq.defCU1}) of the space $\CU^{(1)}$ is more complicated to interpret, we will see in the Corollary \ref{QCorFinDimSubspaces} that it can also be defined as
$$
 \CU^{(1)}=(\BGG_1 \CV \oplus \BGG_2 \CV) \ominus \CV,
$$
in other words, as the orthogonal complement of $\CV$ in $\BGG_1 \CV \oplus \BGG_2 \CV$.
\end{Rem}

The $Y$-operator of the associated $Y$-problem in the $Y(2)$ subspace collection (\ref{QYSubspaceCollection}) will be the $Y$-tensor function $\BYE(L_1,L_2)$ and the effective tensor function of the $Z(2)$-subspace collection (\ref{QZ1SubspaceCollection}) will be $\BLE^{(1)}(L_1,L_2)$, i.e., the $Z$-operator%
\index{Zo@$Z$-operator}
of the associated $Z$-problem%
\index{Zp@$Z$-problem}
in the $Z(2)$ subspace collection (\ref{QZ1SubspaceCollection}). The purpose of this subsection is to define these two operators and prove that they are congruent operators.%
\index{operators!congruent}

Following Section 7.2 of \citeAY{Graeme:2016:ETC},
associated to this $Y(2)$ subspace collection (\ref{QYSubspaceCollection}) is the $Y$-tensor function (i.e., the associated $Y$-operator)%
\index{Yo@$Y$-operator}
which is a linear operator-valued function $\BYE(L_1,L_2)$ acting on the space $\CV$.
To obtain this function we begin by solving the following problem: for a given $\BE_1\in\CV$, find a unique vector pair: $\BE\in\CE$, $\BJ\in\CJ$ that satisfy
\begin{align}\label{QYProblem}
\BJ_2=\BL\BE_2,\;\;\BE_2=\BGP_2\BE,\;\;\BJ_2=\BGP_2\BJ,\;\;\BE_1=\BGP_1\BE,
\end{align}
also known as the $Y$-problem%
\index{Yp@$Y$-problem}
for the $Y(2)$ subspace collection (\ref{QYSubspaceCollection}) (see Chapter 7 of \citeAY{Graeme:2016:ETC}).
The associated operator $\BYE$, by definition, governs the linear relation
\begin{align}\label{QYOperator}
\BJ_1=-\BYE\BE_1,
\end{align}
where $\BJ_1=\BGP_1\BJ$.
The formula for the operator $\BYE$, as given in equation (19.29) of \citeAPY{Milton:2002:TOC} [see also Section 7.4,
equation (7.60), of \citeAY{Graeme:2016:ETC}] is%
\index{effective!operator!Y@$Y_*$ formula}
\begin{align}\label{Qeq.defeffectiveY}
\BYE=\BGP_1\BGG_2(\BGG_2\BGP_2{\BL}^{-1}\BGP_2\BGG_2)^{-1}\BGG_2\BGP_1,
\end{align}
where the inverse $(\BGG_2\BGP_2{\BL}^{-1}\BGP_2\BGG_2)^{-1}$, if it exists, is to be taken on the subspace $\CJ$. For instance, if $\Imag\BL\ge\alpha\mathbb{I}$ or $\Real\BL\ge\alpha\mathbb{I}$ for some $\alpha>0$ and Assumption \ref{QAsmProdProjCoersive} below is satisfied, then the inverse of $\BGG_2\BGP_2{\BL}^{-1}\BGP_2\BGG_2$ exists on $\CJ$.
\begin{Rem}\label{QRemLCommsWithPi2}
It follows from Remark \ref{QRemLOnSubspaceRestr} that
\begin{align}
\BGP_2\BL\BGP_2=\BGP_2\BL=\BL\BGP_2 \,\,\, \mbox{on $\CK$}.
\end{align}
which implies that if $~\BL^{-1}$ exists then
\begin{align}\label{Qeq.cominvpi2}
\BGP_2\BL^{-1}\BGP_2=\BGP_2\BL^{-1}=\BL^{-1}\BGP_2 \,\,\, \mbox{on $\CK$}.
\end{align}
\end{Rem}

\begin{Asm}\label{QAsmProdProjCoersive}
We will assume that
\begin{align}
\BGG_2\BGP_2\BGG_2\geq \beta \BGG_2\;\;\textnormal{for some }\beta>0.
\end{align}
\end{Asm}
\begin{Pro}\label{QPropNecSufCondAssump}
A necessary and sufficient condition for Assumption \ref{QAsmProdProjCoersive} to be true is
\begin{align}
\CV\cap\CJ=\{{\bf 0}\}.
\end{align}
\end{Pro}
\begin{proof}
For any $\BP\in\CV\cap\CJ$ we have $\BGG_2\BP=\BP$ and $\BGP_2\BP=0$ so that $\BGG_2\BGP_2\BGG_2\BP=0$. This proves that a necessary condition for Assumption \ref{QAsmProdProjCoersive} to be true is $\CV\cap\CJ=\{{\bf 0}\}$.  We will now prove that it is also a sufficient condition. The key observations are that $\mathbb{I}_{\CK}=\BGG_1+\BGG_2=\BGP_1+\BGP_2$ which implies
\begin{align}
0\leq\BGG_2\BGP_2\BGG_2=\BGG_2(\mathbb{I}_{\CK}-\BGP_1)\BGG_2=\BGG_2-\BGG_2\BGP_1\BGG_2,
\end{align}
and $\BGP_1$ is a finite-rank operator since its range is $\CV$ which is a finite-dimensional subspace by Proposition \ref{QPropFinDimSpaceSep}. This implies that
\begin{align}
0\leq\BGG_2\BGP_2\BGG_2|_{\CJ}=\mathbb{I}_{\CJ}-\BGG_2\BGP_1\BGG_2|_{\CJ},
\end{align}
where $\BGG_2\BGP_1\BGG_2|_{\CJ}$ is a finite-rank operator on $\CJ$, the range of $\BGG_2$. This implies by Proposition \ref{QPropFredholmSumEquiv} that $\BGG_2\BGP_2\BGG_2|_{\CJ}$ is an index $0$ Fredholm operator%
\index{Fredholm operator!index $0$}
on $\CJ$. Thus, either $\BGG_2\BGP_2\BGG_2|_{\CJ}\BP=0$, $\BP\in\CJ$ has a nontrivial solution or $\BGG_2\BGP_2\BGG_2|_{\CJ}$ is invertible on $\CJ$. But $\BGG_2\BGP_2\BGG_2|_{\CJ}\BP=0$, $\BP\in\CJ$ holds if and only if $\BP\in\CV\cap\CJ$. This implies $\BGG_2\BGP_2\BGG_2|_{\CJ}$ is invertible if and only if $\CV\cap\CJ=\{{\bf 0}\}$. This proves that $\BGG_2\BGP_2\BGG_2>0$ on $\CJ$ if and only if $\CV\cap\CJ=\{{\bf 0}\}$. We now want to prove that we actually have the stronger result, namely, if $\CV\cap\CJ=\{{\bf 0}\}$ then Assumption \ref{QAsmProdProjCoersive} is true. This follows from immediately from the fact $\BGG_2\BGP_2\BGG_2|_{\CJ}$ is a positive semidefinite self-adjoint bounded operator (its norm is bounded by $1$) which is invertible. Hence, the spectrum $\sigma({\BGG_2\BGP_2\BGG_2|_{\CJ}})$ of $\BGG_2\BGP_2\BGG_2|_{\CJ}$ is contained in the closure of its
numerical range which is a real convex compact set:
$$
[\alpha,\beta]=\operatorname{cl}{ \{ (\BGG_2\BGP_2\BGG_2 \,\Bu, \Bu) \mid \Bu \in \CJ,\;\;\|\Bu\|=1 \} },
$$
and this interval $[\alpha,\beta]\subseteq [0, 1]$ has its endpoints $\alpha$ and $\beta$ belonging to $\sigma({\BGG_2\BGP_2\BGG_2|_{\CJ}})$. Thus, as the operator $\BGG_2\BGP_2\BGG_2|_{\CJ}$ is invertible, $\alpha $ must be positive and therefore we get by the definition of the numerical range that $\BGG_2\BGP_2\BGG_2|_{\CJ} \geq \alpha \, \mathbb{I}_{\CJ}$.
\end{proof}

\begin{Cor}\label{Qcor.defYE}
Define the operator-valued function $\BF(L_1,L_2):\CJ\rightarrow\CJ$ on $\bbC^2$ by
$$
\BF(L_1,L_2):=\BGG_2\BGP_2{\BL}(L_1,L_2)^{-1}\BGP_2\BGG_2.
$$
If $\CV\cap\CJ=\{{\bf 0}\}$ then $\operatorname{Ker}[\BF(L_1,L_2)]=\operatorname{Ker}[\BF(L_1,L_2)^{\dag}]=\{{\bf 0}\}$ and the following three statements are equivalent: \newline
\indent (i) $\BYE(L_1,L_2)$ is well-defined by the formula (\ref{Qeq.defeffectiveY});
\newline
\indent (ii) $\operatorname{Ran}[\BF(L_1,L_2)]=\CJ$;
\newline
\indent (iii) $\BF(L_1,L_2)$ is a Fredholm operator of index $0$.
\end{Cor}
\begin{proof}
Suppose $\CV\cap\CJ=\{{\bf 0}\}$. We will now show that $\operatorname{Ker}[\BF(L_1,L_2)]=\{{\bf 0}\}$. Let $\BJ\in \operatorname{Ker}[\BF(L_1,L_2)]$. Then $\BJ\in \CJ$ and ${\bf 0}=\BF(L_1,L_2)\BJ=\BGG_2\BGP_2\BL(L_1,L_2)^{-1}\BGP_2\BJ$ implying that ${\bf 0}=\BGP_2\BGG_2\BGP_2\BL(L_1,L_2)^{-1}\BGP_2\BJ$. But $\BGG_2\BGP_2\BGG_2$ is invertible on $\CJ$ (as we showed in the proof of Proposition \ref{QPropNecSufCondAssump}) and $\BGP_2$ commutes with $\BL(L_1,L_2)$ and hence with ${\BL}(L_1,L_2)^{-1}$ implying ${\bf 0}=\BGP_2\BL(L_1,L_2)^{-1}\BGP_2\BJ=\BL(L_1,L_2)^{-1}\BGP_2\BJ$. Therefore, we conclude that $\BGP_2\BJ={\bf 0}$ which implies $\BJ\in \CV\cap\CJ=\{{\bf 0}\}$ and hence $\operatorname{Ker}[\BF(L_1,L_2)]=\{{\bf 0}\}$. As $\BF(L_1,L_2)^{\dag}=\BF(\overline{L_1},\overline{L_2})$, this implies $\operatorname{Ker}[\BF(L_1,L_2)^{\dag}]=\{{\bf 0}\}$. The proof of the rest of the statement now follows immediately from these facts and the fact that since $\CJ$ is a Hilbert space then $\operatorname{Ker}[\BF(L_1,L_2)^{\dag}]^\perp=\operatorname{cl}\{\operatorname{Ran}[\BF(L_1,L_2)]\}$.
\end{proof}

Using this corollary we can now give, in the next proposition, a condition on $\BL$ that implies $\BYE$ is well-defined by the formula (\ref{QYOperator}).
\begin{Pro}\label{Qprop.defYE}
If $(L_1,L_2)\in \bbC^2$ is such that $\BL(L_1,L_2)$ is invertible and the operator $\BG(L_1,L_2):=(\BGG_0+\BGG_2)\BL(L_1,L_2)^{-1}(\BGG_0+\BGG_2):\CU\oplus\CJ\rightarrow\CU\oplus\CJ$ is also invertible then the operator $\BF(L_1,L_2):\CJ\rightarrow\CJ$ (defined in Corollary \ref{Qcor.defYE}) is a Fredholm operator of index $0$. Moreover, if $\CV\cap\CJ=\{{\bf 0}\}$ then $\BF(L_1,L_2)$ is invertible and $\BYE(L_1,L_2)$ is well-defined by formula (\ref{Qeq.defeffectiveY}).
\end{Pro}
\begin{proof}
Suppose that $(L_1,L_2)\in \bbC^2$ such that $\BL(L_1,L_2)$ is invertible and the operator $\BG(L_1,L_2):\CU\oplus\CJ\rightarrow\CU\oplus\CJ$ defined above is also invertible. Then with respect to the orthogonal decomposition of the Hilbert space $\CU\oplus\CJ$, with the corresponding orthogonal projections $\BGG_0$ and $\BGG_2$, we can write this bounded linear operator in the $2\times 2$ block matrix form
\begin{align}
\BG(L_1,L_2)=
\begin{bmatrix}
\BG_{11} & \BG_{12}\\
\BG_{21} & \BG_{22}
\end{bmatrix},
\end{align}
where $\BG_{22}:\CJ\rightarrow \CJ$ is the operator
\begin{align}
\BG_{22}=\BGG_2\BG(L_1,L_2)\BGG_2=\BGG_2\BL(L_1,L_2)^{-1}\BGG_2.
\end{align}
It follows immediately from Proposition \ref{QPropFredholmBlockEquiv} that $\BG_{22}$ is a Fredholm operator of index $0$. Using the relation (\ref{Qeq.cominvpi2}) and the decomposition $\BGG_2=\BGP_1\BGG_2+\BGP_2\BGG_2$, we can link $\BG_{22}$ and $\BF(L_1,L_2)$ by the following relation:
\begin{align}
\BF(L_1,L_2)&=\BGG_2\BGP_2\BL(L_1,L_2)^{-1}\BGP_2\BGG_2\\
&=\BG_{22}-\BGG_2\BL(L_1,L_2)^{-1}\BGP_1\BGG_2.
\end{align}
Thus, as $\operatorname{Ran}\BGP_1=\CV$ is a finite-dimensional space (see Corollary \ref{QCorFinDimSubspaces}), by the invariance of the Fredholm index under perturbation by a compact operator (see Theorem 4.1, p.\ 355, section 15.4 of \citeAY{Gohberg:2003:BCO}), we obtain that $\BF(L_1,L_2)$ is a Fredholm operator of index $0$. The rest of the statements of this proposition now follow immediately from these facts and Corollary \ref{Qcor.defYE}. This completes the proof.
\end{proof}

We will now show that just as $\BLE$ satisfies the homogeneity and Herglotz property so too does $\BYE$ under the Assumption \ref{QAsmProdProjCoersive}, that is, we prove the following:

{\bf The homogeneity property:}
\begin{align}
\BYE(L_1,L_2)=\frac{1}{c}\BYE(cL_1,cL_2)
\end{align}
is satisfied for all choices of constants $c\not=0$.

{\bf The Herglotz property:}
\begin{align}\label{Qeq.HergL}
\Imag\BYE(L_1,L_2)>0,\;\;\textnormal{when } \Imag(L_1)>0\textnormal{ and }\Imag(L_2)>0.
\end{align}

\begin{proof}
We will prove these two properties under Assumption \ref{QAsmProdProjCoersive}. First, the homogeneity property follows immediately from formula (\ref{Qeq.defeffectiveY}) and the homogeneity property of $\BLE$.  Next, the  Herglotz  property of $\BYE$ will be proven in the same as for $\BLE$ by using formula (\ref{Qeq.defeffectiveY}) and the identities (\ref{QFundImProps}).
[Alternatively, the Herglotz property of $\BYE$ follows directly from Eq.\ (2.82) of Chapter 2 of \citeAY{Graeme:2016:ETC} in the same way as the Herglotz property of $\BLE$ follows from Eq.\ (2.68) of Chapter 2 of \citeAY{Graeme:2016:ETC} or \eq{QLposit}.]  First, we obtain that under the Herglotz conditions $ \Imag(L_1)>0$ and $\Imag(L_2)>0$, that $\Imag(\BL^{-1})$ is uniformly negative. Then by Assumption \ref{QAsmProdProjCoersive}, it follows that the operator $\Imag(\BGG_2\BGP_2{\BL}^{-1}\BGP_2\BGG_2)$ is uniformly negative on $\CJ$ and that $\Imag(\BGG_2(\BGG_2\BGP_2{\BL}^{-1}\BGP_2\BGG_2)^{-1}\BGG_2)$ is coercive on $\CJ$. From these facts we conclude that $$\Imag(\BGP_1\BGG_2(\BGG_2\BGP_2{\BL}^{-1}\BGP_2\BGG_2)^{-1}\BGG_2\BGP_1)>0 \ \mbox{ on $\CV$.}$$
\end{proof}

The relationship between $\BLE$ and $\BYE$, proven in Section 19.1 of \citeAPY{Milton:2002:TOC} and for two-component composites and in the abstract setting in Section 7.21 of \citeAY{Graeme:2016:ETC},
 is given by
\begin{align}\label{QRelBetweenLAndYEffTens}
\BLE = \BGG_0\BL\BGG_0-\BGG_0\BL\BGP_1[\BGP_1\BL\BGP_1+\BYE]^{-1}\BGP_1\BL\BGG_0.
\end{align}
\begin{proof}
For every $\Be\in \CU$ which has a solution to the $Z$-problem, there exists $\Bj\in \CU$, $\BE\in\CE$, and $\BJ\in\CJ$ such that
\begin{align}
\Bj+\BJ=\BL(\Be+\BE),\;\;\BLE\Be=\Bj,
\end{align}
the latter by definition of $\BLE$. It follows from Remark \ref{QRemLCommsWithPi2} that $\BE_1=\BGP_1\BE\in\CV$ is a solution to the $Y$-problem (\ref{QYProblem}), i.e.,
\begin{align}
\BJ_2=\BL\BE_2,\;\;\BE_2=\BGP_2\BE,\;\;\BJ_2=\BGP_2\BJ,\;\;\BE_1=\BGP_1\BE,
\end{align}
and, by definition we have
\begin{align}
\BJ_1=-\BYE\BE_1,\;\;\BJ_1=\BGP_1\BJ.
\end{align}
Thus we can write
\begin{align}
\Bj+\BJ_1+\BJ_2=\BL(\Be+\BE_1+\BE_2).
\end{align}
We will now prove from these facts that the relationship (\ref{QRelBetweenLAndYEffTens}) holds. First, we solve for $\BE_1$ in terms of $\Be$ from the identity
\begin{align}
-\BYE\BE_1=\BJ_1=\BGP_1\BL(\Be+\BE_1)=\BGP_1\BL(\Be+\BGP_1\BE_1)
\end{align}
and we find that
\begin{align}
\BE_1=-\BGP_1(\BGP_1\BL\BGP_1+\BYE)^{-1}\BGP_1\BL\Be,
\end{align}
where the inverse is taken on the subspace $\CV$. It then follows that
\begin{align}
\Bj+\BJ_1=\BL(\Be+\BE_1)=\BL\Be+\BL\BE_1 =\BL\Be-\BL\BGP_1(\BGP_1\BL\BGP_1+\BYE)^{-1}\BGP_1\BL\Be.
\end{align}
Hence we have
\begin{align}
\BLE\Be=\Bj=\BGG_0(\Bj+\BJ_1)=[\BGG_0\BL\BGG_0-\BGG_0\BL\BGP_1(\BGP_1\BL\BGP_1+\BYE)^{-1}\BGP_1\BL\BGG_0]\Be.
\end{align}
Therefore, it follows that the relationship (\ref{QRelBetweenLAndYEffTens}) is true if $\BLE$, $\BYE$, and $(\BGP_1\BL\BGP_1+\BYE)^{-1}$ are well-defined operators (which is true, for instance, when $\Imag(L_1)>0$, $\Imag(L_2)>0$, and Assumption \ref{QAsmProdProjCoersive} holds).
\end{proof}

Now consider the Hilbert space $\CH^{(1)}$, i.e., the orthogonal $Z(2)$ subspace collection (\ref{QZ1SubspaceCollection}). First, we emphasize here that the operator $\BL$ commutes with the orthogonal projection $\BGP_2$ whose range is $\CH^{(1)}$ and hence by Remark \ref{QRemLOnSubspaceRestr} it follows that
\begin{align}
\BL=\BGP_2\BL\BGP_2=L_1\BGL^{(1)}_1+L_2\BGL^{(1)}_2 \textnormal{ on }\CH^{(1)},
\end{align}
where now $\BGL^{(1)}_1$ and $\BGL^{(1)}_2$ on $\CH^{(1)}$ are the orthogonal projections onto $\CP_1^{(1)}$ and $\CP_2^{(1)}$, respectively.
Now associated with the operator $\BL$ on $\CH^{(1)}$ and this orthogonal $Z(2)$ subspace collection is the linear operator-valued function $\BLE^{(1)}(L_1,L_2)$ (i.e., the associated $Z$-operator) acting on the subspace $\CU^{(1)}$ which is the solution of the corresponding $Z$-problem in this new $Z(2)$ subspace collection. The results in section \ref{QSecFormProb} now apply to this operator $\BLE^{(1)}(L_1,L_2)$ which, just like $\BLE$, satisfies the homogeneity, Herglotz, and normalization properties%
\index{homogeneity property}%
\index{Herglotz property}%
\index{normalization property}
[i.e., the normalization property in this case is $\BLE^{(1)}(1,1)=\mathbb{I}_{\CU^{(1)}}$]. In particular, $\BLE^{(1)}$  admits the two following representation formulas (when they are well-defined) similar to the formulas and  (\ref{QZOperatorFormulaschur}) (\ref{QZOperatorFormula}) but now on the orthogonal $Z(2)$ subspace collection (\ref{QZ1SubspaceCollection}):
\begin{align}\label{QZ1OperatorSchur}
\BLE^{(1)}=\BGG_0^{(1)}\BL \BGG_0^{(1)}- \BGG_0^{(1)}\BL\BGG_1^{(1)}(\BGG_1^{(1)}\BL \BGG_1^{(1)})^{-1} \BGG_1^{(1)}\BL \BGG_0^{(1)}
\end{align}
and
\begin{align}\label{QZ1Operator}
\BLE^{(1)}=\BGG_0^{(1)}[(\BGG_0^{(1)}+\BGG_2^{(1)})\BL^{-1}(\BGG_0^{(1)}+\BGG_2^{(1)})]^{-1}\BGG_0^{(1)},
\end{align}
where $\BGG_0^{(1)}, \BGG_1^{(1)}$  and $\BGG_2^{(2)}$ are respectively the orthogonal projections onto $\CU^{(1)}$, $\CE^{(1)}$, $\CJ^{(1)}$,  which are related to the orthogonal projection $\BGP_2$ onto $\CH^{(1)}$ by $\BGP_2=\BGG_0^{(1)}+\BGG_1^{(2)}+\BGG_2^{(2)}$.
Then, can one link the operators $\BYE$ associated with $\CV$ and $\BLE^{(1)}$ associated with $\CU^{(1)}$? The answer is yes.
The first step in the proof is to show that these two operators act on finite-dimensional spaces of the same dimension, in other words that $\dim \CU^{(1)}=\dim(\CV)$. This is a consequence of the following proposition and corollary.
\begin{Pro}\label{QPropFinDimSpaceSep}
Let $\CL$ be a Hilbert space with two orthogonal projections $\BCQ_1$, $\BCQ_2$ which satisfy $\BCQ_1+\BCQ_2=\mathbb{I}_{\CL}$. Assume that $\CL$ has an orthogonal decomposition
\begin{align}
\CL=\CM\oplus\CN.
\end{align}
Define the subspaces $\CN_1,\CN_2$ by
\begin{align}
\CN_1=(\BCQ_1\CL\cap\CN)\oplus(\BCQ_2\CL\cap\CN),\;\;\CN_2=\CN\ominus\CN_1,
\end{align}
where the latter denotes the orthogonal complement of $\CN_1$ in $\CN$.
Then we have
\begin{align}\label{Qeq.spacedecomp}
\CM\oplus\CN_2=\BCQ_1\CM\oplus\BCQ_2\CM.
\end{align}
In particular, if $\CM$ is finite-dimensional then $\CN_2$ is finite-dimensional with
\begin{align}\label{Qeq.spacedim}
\dim(\CN_2)=\dim(\BCQ_1\CM)+\dim(\BCQ_2\CM)-\dim(\CM)\leq \dim(\CM).\end{align}
Moreover, $\dim(\CN_2)=\dim(\CM)$ if and only if $\BCQ_1\CL\cap\CM=\BCQ_2\CL\cap\CM=\{{\bf 0}\}$.
\end{Pro}
\begin{proof}
First, $\CN_1\subseteq (\BCQ_1\CM\oplus\BCQ_2\CM)^\perp$ holds since $\BCQ_1,\BCQ_2$ are invariant on $\CN_1$ and $\CM$ is orthogonal to $\CN$ which contains $\CN_1$. This implies $\CM\subseteq\BCQ_1\CM\oplus\BCQ_2\CM \subseteq (\BCQ_1\CM\oplus\BCQ_2\CM)^{\perp\perp}\subseteq \CN_1^\perp=\CM\oplus\CN_2$. Let $\BP\in\CM\oplus\CN_2\ominus(\BCQ_1\CM\oplus\BCQ_2\CM)$. Then $\BP\in (\BCQ_1\CM\oplus\BCQ_2\CM)^\perp\subseteq \CM^\perp=\CN_1\oplus\CN_2$ implying $\BP\in \CN_2$. But by invariance of $\BCQ_1, \BCQ_2$ on $\BCQ_1\CM\oplus\BCQ_2\CM$ we must also have $\BCQ_1\BP, \BCQ_2\BP\in (\BCQ_1\CM\oplus\BCQ_2\CM)^\perp\subseteq \CN$ and thus $\BCQ_1\BP, \BCQ_2\BP\in\CN_1$. This implies $\BP=\BCQ_1\BP+\BCQ_2\BP\in \CN_1\cap\CN_2=\{{\bf 0}\}$ and hence $\BP={\bf 0}$. Therefore, $\CM\oplus\CN_2\ominus(\BCQ_1\CM\oplus\BCQ_2\CM)=\{{\bf 0}\}$ implying the relation (\ref{Qeq.spacedecomp}): $\CM\oplus\CN_2=\BCQ_1\CM\oplus\BCQ_2\CM$, as desired. Now assume that $\CM$ is a finite-dimensional space, then the relation (\ref{Qeq.spacedim}) follows immediately from the decomposition (\ref{Qeq.spacedecomp}). Next, the equality $\dim(\CN_2)=\dim(\CM)$ is equivalent to $\dim(\BQ_1 \CM)=\dim(\BQ_2 \CM)=\dim(\CM)$ and, by the rank theorem, it is equivalent to the injectivity of the restriction of $\BQ_1$ and $\BQ_2$ on $\CM$.  Using the fact $\BQ_1$ and $\BQ_2$ are two projections which satisfy $\BCQ_1+\BCQ_2=\mathbb{I}_{\CL}$, one can easily show that this injectivity condition is equivalent to $\BCQ_1\CL\cap\CM=\BCQ_2\CL\cap\CM=\{{\bf 0}\}$. This completes the proof.
\end{proof}

\begin{Cor}\label{QCorFinDimSubspaces}
The spaces $\CV$ and $\CU^{(1)}$ are finite-dimensional and orthogonal to $\CU$ and $\CV$, respectively. Furthermore, we have
\begin{align}
\CU\oplus\CV = \BGL_1\CU\oplus\BGL_2\CU,\\
\CV\oplus\CU^{(1)} = \BGG_1\CV\oplus\BGG_2\CV,
\end{align}
and
\begin{align}
\dim(\CV)=\dim(\BGL_1\CU)+\dim(\BGL_2\CU)-\dim(\CU)\leq\dim(\CU),\\
\dim(\CU^{(1)})=\dim(\BGG_1\CV)+\dim(\BGG_2\CV)-\dim(\CV)\leq\dim(\CV).
\end{align}
Moreover, the following two statements are true: \newline
\indent (i) $\dim(\CV)=\dim(\CU)$ if and only if $\CP_1\cap\CU=\CP_2\cap\CU=\{{\bf 0}\}$; \newline
\indent (ii) $\dim(\CU^{(1)})=\dim(\CV)$ if and only if $\CE\cap\CV=\CJ\cap\CV=\{{\bf 0}\}$.
\end{Cor}
\begin{proof}
These results follow immediately from Proposition \ref{QPropFinDimSpaceSep}. Indeed, in the Hilbert space $\CH$ we have the orthogonal decomposition $\CH=\CU\oplus\CK$ and two orthogonal projections $\BGL_1, \BGL_2$ satisfying $\BGL_1+\BGL_2=\mathbb{I}_{\CH}$. The spaces $\CL, \CM, \CN, \CN_1, \CN_2$ in the above proposition are $\CH, \CU, \CK, \CH^{(1)}=\CP_1^{(1)}\oplus\CP_2^{(1)}, \CV$, respectively, since $\BGL_1\CH=\CP_1, \BGL_2\CH=\CP_2$. Similarly, on the Hilbert space $\CK$ we have $\CK=\CV\oplus\CH^{(1)}$ and there are two orthogonal projections $\BGG_1, \BGG_2$ satisfying $\BGG_1+\BGG_2=\mathbb{I}_{\CK}$. The spaces $\CL, \CM, \CN, \CN_1, \CN_2$ in the above proposition are $\CK, \CV, \CH^{(1)}, \CH^{(1)}\ominus \CU^{(1)}, \CU^{(1)}$, respectively, since $\BGG_1\CK=\CE, \BGG_2\CK=\CJ$. The corollary now follows from these identifications by Proposition \ref{QPropFinDimSpaceSep}. This completes the proof.
\end{proof}

\begin{Rem}
In the case of the conductivity equation for a two-component composite, the conditions $\CP_1\cap\CU=\CP_2\cap\CU=\{{\bf 0}\}$ always hold. This is a direct consequence of the definition of the spaces $\CP_1, \CP_2$ and $\CU$ (see Remarks \ref{Qrem.conduct1} and \ref{Qrem.conduct2}) and we deduce that $\dim \CV=\dim \CU=d$.
Concerning the conditions $\CV\cap\CJ=\{{\bf 0}\}$ and $\CV\cap\CE=\{{\bf 0}\}$, they are nearly always satisfied. Indeed,
it depends on the geometry shape of the interfaces between the two phases. Suppose that the two phases of the composite are Lipschitz domains. A
current density $\BJ \in \CV\cap\CJ$ is (by definition of $\CV$) a constant $\BC_i$ in each phase. Moreover, it belongs to $H_{\Divop}(\CD)=\{ \BU \in \BL^{2}(\CD) \mid \Divop \BU\in \BL^{2}(\CD) \}$. Thus at each interface point between the two phases, the normal component of $\BJ$ has to be continuous (see chapter I of \citeAY{Monk:2003:FEM}):
$$
[\BJ \cdot \Bn]=(\BC_2-\BC_1) \cdot \Bn ={\bf 0},
$$
where $\Bn$ denotes the unit normal vector (oriented for instance from phase 1 to phase 2). Thus, if the set of unit normal of all interface points contains $d$ linear independent vectors, then $\BJ=\BC_2=\BC_1$ in $\CD$ and the zero average condition on $\BJ$ then implies $\BJ={\bf 0}$.
In the same way, if an electrical field $\BE \in \CV\cap\CE$, it has to be a constant $\BC_i$ in each phase and it belongs to $H_{\Curlop}(\CD)=\{ \BU \in \BL^{2}(\CD)\mid \Curlop \BU\in \BL^{2}(\CD) \}$ which implies that its tangential component is continuous (see chapter I of \citeAY{Monk:2003:FEM}). This leads to
$$
[\BE \wedge \Bn]=(\BC_2-\BC_1) \wedge \Bn ={\bf 0},
$$
at each interface point between the two phases. Thus again, if  the set of unit normal of all interface points contains $d$ linear independent vectors, then $\BE=\BC_2=\BC_1$ in $\CD$ and the zero average condition on $\BE$ implies that  $\BE={\bf 0}$. Nevertheless, one can construct easily for $d\geq 2$ counterexamples (with for example linear interfaces)
where $\CV\cap\CJ \neq \{{\bf 0}\}$ or $\CV\cap\CE \neq \{{\bf 0}\}$.

For instance, consider the example from \citeAY{Milton:2002:TOC}, pp.\ 407--408, section 19.4 of a laminate of two phases laminated in direction $\Bn$: Let $f_1,f_2$ denote the volume fraction occupied by phase $i=1,2$, then the piecewise constant average value zero fields
\begin{align}
\BE=[f_2\chi_1-f_1\chi_2]\Bn,\;\;\BJ=[f_2\chi_1-f_1\chi_2]\Bv,\textit{ with }\Bn\cdot\Bv=0,
\end{align}
are curl-free and divergence-free, respectively, and therefore lie in $\CV\cap\CE \neq \{{\bf 0}\}$ and $\CV\cap\CJ \neq \{{\bf 0}\}$, respectively.
\end{Rem}

Now, that we know under some assumptions that $\dim \CV=\dim \CU^{(1)}$, the second step is to construct an invertible linear map $\BK:\CU^{(1)} \to \CV$ to show that the operators $\BYE$ and $\BLE^{(1)}$ are $\BK$ congruent,%
\index{operators!congruent}
in other words that:
$$
\BYE(L_1,L_2)=\BK\BLE^{(1)}(L_1,L_2)\BK^{\dag}.
$$
This is the purpose of the following theorem.
\begin{Thm}\label{Qthm.relYStarAndLStar1}
If $\CP_1\cap \CU=\CP_2\cap \CU=\{{\bf 0}\}$ and $\CE\cap\CV=\CJ\cap\CV=\{{\bf 0}\}$ then Assumption \ref{QAsmProdProjCoersive} holds, $\dim \CU=\dim \CV=\dim \CU^{(1)}$, and assuming the operators $\BYE=\BYE(L_1,L_2)$, $\BLE^{(1)}=\BLE^{(1)}(L_1,L_2)$ operators in (\ref{Qeq.defeffectiveY}) and (\ref{QZ1OperatorSchur}), respectively, are well-defined (which is the case, for instance, if $\BL$ is coercive) we have
\begin{align}\label{QDefNextEffTensor}
\BYE(L_1,L_2)=\BK\BLE^{(1)}(L_1,L_2)\BK^{\dag},
\end{align}
where $\BK:\CU^{(1)}\rightarrow \CV$ is the invertible operator defined by
\begin{align}\label{QDefKMap}
\BK=-(\BGP_1\BGG_1\BGP_1)^{-1}\BGP_1\BGG_1\BGP_2,
\end{align}
with the inverse of $\BGP_1\BGG_1\BGP_1$ taken on $\CV$.
\end{Thm}
\begin{proof}
Suppose that $\CP_1\cap \CU=\CP_2\cap \CU=\{{\bf 0}\}$, $\CE\cap\CV=\CJ\cap\CV=\{{\bf 0}\}$. It then follows that by Proposition \ref{QAsmProdProjCoersive} since $\CJ\cap\CV=\{{\bf 0}\}$ we know that Assumption \ref{QAsmProdProjCoersive} holds, by Corollary \ref{QCorFinDimSubspaces} since $\CP_1\cap \CU=\CP_2\cap \CU=\{{\bf 0}\}$ we know that $\dim \CU=\dim \CV$ and since $\CE\cap\CV=\CJ\cap\CV=\{{\bf 0}\}$ we know that $\dim \CV=\dim \CU^{(1)}$. Also, since $\CJ\cap\CV=\{{\bf 0}\}$ and $\CV$ is finite-dimensional then $\BGP_1\BGG_1\BGP_1$ is invertible on $\CV$ and hence the operator $\BK:\CU^{(1)}\rightarrow \CV$ given by (\ref{QDefKMap}) is well-defined. We will now prove it is invertible. To do this we introduce the operator $\BK^{\prime}:\CU^{(1)}\rightarrow \CV$ defined by
\begin{align}\label{QDefKPrimeMap}
\BK^{\prime}=-(\BGP_1\BGG_2\BGP_1)^{-1}\BGP_1\BGG_2\BGP_2,
\end{align}
with the inverse of $\BGP_1\BGG_2\BGP_1$ taken on $\CV$, which exists since by assumption $\CE\cap\CV=\{{\bf 0}\}$ and $\CV$ is finite-dimensional.  We will now prove that $\BK$ is invertible by showing that $\BK^{\dag}=-(\BK^{\prime})^{-1}$. This follows now from the fact that $\dim \CV=\dim \CU^{(1)}<\infty$ and on $\CV$,
\begin{gather*}
\BK^{\prime}\BK^{\dag}=(\BGP_1\BGG_2\BGP_1)^{-1}\BGP_1\BGG_2\BGP_2\BGP_2\BGG_1\BGP_1(\BGP_1\BGG_1\BGP_1)^{-1}\\
=(\BGP_1\BGG_2\BGP_1)^{-1}\BGP_1\BGG_2(\mathbb{I}_{\CK}-\BGP_1)\BGG_1\BGP_1(\BGP_1\BGG_1\BGP_1)^{-1}\\
=-(\BGP_1\BGG_2\BGP_1)^{-1}\BGP_1\BGG_2\BGP_1\BGG_1\BGP_1(\BGP_1\BGG_1\BGP_1)^{-1}=-\BGP_1\BGG_1\BGP_1(\BGP_1\BGG_1\BGP_1)^{-1}=-\mathbb{I}_{\CV}.
\end{gather*}
Now suppose the operators $\BYE=\BYE(L_1,L_2)$, $\BLE^{(1)}=\BLE^{(1)}(L_1,L_2)$ in (\ref{Qeq.defeffectiveY}) and (\ref{QZ1Operator}), respectively, are well-defined. We will prove the identity (\ref{QDefNextEffTensor}). For any $\BE_1\in \CV$ we have $\BJ_1:=-\BYE\BE_1\in\CV$ and we know that there exists a solution to the $Y$-problem, i.e., there exists an $\BE^{\prime}\in\CE$ and $\BJ^{\prime}\in\CJ$ such that $\BGP_1\BE^{\prime}=\BE_1$, $\BGP_1\BJ^{\prime}=\BJ_1$ and $\BJ_2=\BL\BE_2$ where $\BJ_2:=\BGP_2\BJ^{\prime}$ and $\BE_2:=\BGP_2\BE^{\prime}$. Moreover, we remark that
\begin{align*}
\BE_2=\BGP_2\BE^{\prime} \in \BGP_2 (\CE) \subseteq  \BGP_2 (\CE_1\oplus \CU_1\oplus \CV)=\CE_1\oplus \CU_1,
\end{align*}
and in the same way that $\BJ_2 \in \CJ_1\oplus \CU_1$.

But for this $\BE_0^{(1)}:=\BGG_0^{(1)}\BE^{\prime}$ we can construct a solution to the $Z$-problem on the $Z(2)$ subspace collection $\CH^{(1)}=\CU^{(1)}\oplus \CE^{(1)}\oplus\CJ^{(1)}=\CP_1^{(1)}\oplus\CP_2^{(1)}$, namely given $\BE_0^{(1)}\in\CU^{(1)}$ we see that the fields
\begin{align*}
 \BJ_0^{(1)}:=\BGG_0^{(1)}\BJ^{\prime}\in \CU^{(1)},\quad \BE:=\BE_2-\BE_0^{(1)}\in\CE^{(1)},\quad \BJ:=\BJ_2-\BJ_0^{(1)}\in\CJ^{(1)}
\end{align*}
are such that $\BJ_0^{(1)}+\BJ=\BJ_2=\BL(\BE_2)=\BL(\BE_0^{(1)}+\BE)$. But the $Z$-operator associated with the $Z$-problem on this $Z(2)$ subspace collection $\CH^{(1)}$ is $\BLE^{(1)}$ and hence
\begin{align}
\BJ_0^{(1)}=\BLE^{(1)}\BE_0^{(1)}.
\end{align}
We will now prove that $\BJ_1=\BK \BJ_0^{(1)}$ and $\BE_1=\BK^{\prime}\BE_0^{(1)}$. First, since $\BJ_2-\BJ_0^{(1)}=\BJ\in\CJ^{(1)}\subseteq \CJ$, we have
\begin{align}
{\bf 0} = \BGG_1(\BJ_2-\BJ_0^{(1)})=\BGG_1\BJ_2-\BGG_1\BJ_0^{(1)}
\end{align}
which leads to
\begin{align}
\BGP_1 \BGG_1\BJ_2= \BGP_1 \BGG_1 \BGP_2\BJ_0^{(1)}.
\end{align}
Then, using the relation $\BJ_2=\BJ^{\prime}-\BJ_1$ with $\BJ^{\prime}\in \CJ$ and $\BJ_1\in \CV$, we get
\begin{align}
\BGP_1\BGG_1\BGP_1\BJ_1=-\BGP_1\BGG_1\BGP_2\BJ_0^{(1)},
\end{align}
and it leads to
\begin{align}
\BJ_1=\BK\BJ_0^{(1)}.
\end{align}
Similarly, since $\BE_2-\BE_0^{(1)}=\BE\in\CE^{(1)}\subseteq \CE$ then we have
\begin{align}
{\bf 0} = \BGG_2(\BE_2-\BE_0^{(1)})=\BGG_2\BE_2-\BGG_2\BE_0^{(1)},
\end{align}
and we get that
\begin{align}
\BGP_1\BGG_2\BGP_1\BE_1=-\BGP_1\BGG_2\BGP_2\BE_0^{(1)},
\end{align}
which leads to
\begin{align}
\BE_1=\BK^{\prime}\BE_0^{(1)}.
\end{align}
It follows from these facts that
\begin{gather*}
-\BYE\BE_1=\BJ_1=\BK\BJ_0^{(1)}=\BK\BLE^{(1)}\BE_0^{(1)}=\BK\BLE^{(1)}(\BK^{\prime})^{-1}\BE^{(1)}=-\BK\BLE^{(1)}\BK^{\dag}\BE_1.
\end{gather*}
As this is true for every $\BE_1\in\CV$ it implies the relation
\begin{align}
\BYE=\BK\BLE^{(1)}\BK^{\dag},
\end{align}
as desired. This completes the proof.
\end{proof}

We want now to give necessary conditions under which the operators $\BLE^{(1)}$ is well-defined.
This question is not as clear as for $\BYE$ (see Proposition \ref{Qprop.defYE}) since the well-posedness of the effective tensors $\BLE^{(1)}$ (as for the effective $\BLE$ for the base case) is directly linked to the properties of the operators $\BL$ in the subspace decomposition $\CH^{(1)}=\CU^{(1)}\oplus \CE^{(1)}\oplus \CJ^{(1)}$ and therefore is highly dependent on the physical problem and the structure of the composite.

For instance, the representation formulas  (\ref{QZ1OperatorSchur})  and (\ref{QZ1Operator}) of $\BLE^{(1)}$ are respectively defined under the existence of the inverse $\BGG_1^{(1)}\BL \BGG_1^{(1)}$ on $\CE^{(i)}$ and of the inverse of $(\BGG_0^{(1)}+\BGG_2^{(1)})\BL^{-1}(\BGG_0^{(1)}+\BGG_2^{(1)})$ on $\CU^{(1)}\oplus \CJ^{(1)}$. We saw indeed that the existence  of these both inverses is equivalent as soon as $\BL$ is invertible on $\CH^{(1)}$ (see Proposition \ref{Qprop.equivrepresLE}). By  the Remarks \ref{QRem.inv} and \ref{QRemLOnSubspaceRestr}, the invertibility of  $\BL$ on $\CH^{(1)}$ follows from the invertibility of $\BL$ on $\CH$ which is equivalent to   $L_1\neq 0$ and $L_2\neq 0$. Moreover, if $\BL$ is coercive  (see Remark \ref{QRem.pos})  which is the case for example under the Herglotz hypothesis (\ref{Qeq.HergL}) then $\BLE^{(1)}$ is well-defined by both formulas (\ref{QZ1OperatorSchur}) and (\ref{QZ1Operator}). But what can we say if we don't suppose this strong coercivity assumption?

As the field equation recursion method is an induction method, we want to establish here some criterion of well-posedness for $\BLE^{(1)}$ that it inherits from the well-posedness of $\BLE$ in the base case.
For this purpose, the representation formula (\ref{QZ1OperatorSchur}) is more suitable since the space $\CE^{(1)}$ is constructed as a subspace of $\CE$ (which is not the case of the formula (\ref{QZ1Operator}) since $\CU^{(1)}\oplus \CJ^{(1)}$ is not included in $\CU \oplus \CJ$).
%The goal of the next proposition is to link the wellposdness of $\BLE^{(1)}$ to the wellposdness of $\BLE$ under the representation
%formulas (\ref{QZOperatorFormulaschur}) and (\ref{QZ1OperatorSchur}).
\begin{Pro}\label{QProp.defLEschur}
Let $\CW^{(1)}$ be the subspace of $\CE$ defined by
\begin{align}\label{Qeq.decompCE}
\CW^{(1)}=\CE \ominus  \CE^{(1)},
\end{align}
in other words, the orthogonal complements of $\CE^{(1)}$ in $\CE$.
Then $\CW^{(1)}$ is a finite-dimensional space and $\BL_{11}$ (the restriction of the operator $\BGG_1 \BL\BGG_1$ on $\CE$) can be represented as a $2\times 2$ block operator matrix
\begin{align}
\BL_{11}=\begin{bmatrix}
 \BA & \BB\\
\BC & \BL^{(1)}_{11}
\end{bmatrix},
\end{align}
with respect to the orthogonal decomposition (\ref{Qeq.decompCE}), where $\BL^{(1)}_{11}$ is the restriction of $\BGG_1^{(1)} \BL\BGG_1^{(1)}$ on $\CE^{(1)}$. Furthermore, if $\BL_{11}$ is invertible then $\BL_{11}^{(1)}$ is a index $0$ Fredholm operator%
\index{Fredholm operator!index $0$}
and the invertibility of $\BL_{11}^{(1)}$ is equivalent
to the injectivity condition: $\operatorname{Ker} \BL_{11}^{(1)}=\{ 0 \}$. Hence, if the representation formula (\ref{QZOperatorFormulaschur}) holds for $\BLE$ and  $\operatorname{Ker} \BL_{11}^{(1)}=\{ 0 \}$ then $\BLE^{(1)}$ is well-defined by the formula (\ref{QZ1OperatorSchur}).
\end{Pro}
\begin{proof}
We have simply to prove here that $\CW^{(1)}$ is a finite-dimensional space since all the other conclusions follow immediately from Proposition \ref{QPropFredholmBlockEquiv}.
From the relations (\ref{QYSubspaceCollection}), (\ref{QZ1SubspaceCollection}) and (\ref{QSubspaceInclusionRelations}), we have that
$$
\CK=\CE\oplus \CJ = \CV \oplus \CU^{(1)}\oplus \CE^{(1)}\oplus  \CJ^{(1)} \  \mbox{ with } \CE^{(1)}\subseteq \CE \mbox{ and } \CJ^{(1)}\subseteq \CJ.
$$
Therefore, $\CE=\CE^{(1)} \oplus \CW^{(1)}  \subseteq   \CE^{(1)}\oplus \CV\oplus  \CU^{(1)}$ and, since $\CV$ and $\CU^{(1)}$ are finite dimensional spaces by Corollary \ref{QCorFinDimSubspaces}, this implies that $\CW^{(1)}$ is finite dimensional.
\end{proof}

\subsection{The induction step}\label{QSubsecInductionStep}
We now introduce a hierarchy of subspaces%
\index{subspace collections!subspace hierarchy}
\begin{align}
\CK^{(i-1)}&=\CE^{(i-1)}\oplus\CJ^{(i-1)}=\CV^{(i-1)}\oplus\CP_1^{(i)}\oplus\CP_2^{(i)}, \label{Qeq.subspacecolYi}\\
\CH^{(i)}&=\CU^{(i)}\oplus\CE^{(i)}\oplus\CJ^{(i)}=\CP_1^{(i)}\oplus\CP_2^{(i)}, \label{Qeq.subspacecolZi}
\end{align}
for $i=1,2,3,\ldots,$ where
\begin{align}
\CP_1^{(i)}&=\CP_1\cap\CK^{(i-1)},\;\; \CP_2^{(i)}=\CP_2\cap\CK^{(i-1)},\;\;
\CV^{(i-1)}=\CK^{(i-1)}\ominus(\CP_1^{(i)}\oplus\CP_2^{(i)}),\\
\CE^{(i)}&=\CE\cap\CH^{(i)}, \;\;\CJ^{(i)}=\CJ\cap\CH^{(i)},\;\;\CU^{(i)}=\CH^{(i)}\ominus(\CE^{(i)}\oplus\CJ^{(i)}).
\end{align}
For all positive integers $i$, we denote by $\BGP^{(i)}_1$ and $\BGP^{(i)}_2$ the orthogonal projections on $\CV^{(i-1)}$ and $\CH^{(i)}=\CP_1^{(i)}\oplus\CP_2^{(i)}$ associated with the orthogonal $Y(2)$ subspace collection%
\index{subspace collections!orthogonal}
(\ref{Qeq.subspacecolYi})
and by $\BGG^{(i)}_0$, $\BGG^{(i)}_1$, $\BGG^{(i)}_2$, $\BGL_1^{(i)}$ and $\BGL_2^{(i)}$ the orthogonal projections on $\CU^{(i)}$, $\CE^{(i)}$, $\CJ^{(i)}$, $\CP_1^{(i)}$ and $\CP_2^{(i)}$ associated with the orthogonal $Z(2)$ subspace collection (\ref{Qeq.subspacecolZi}).

The following theorem defines and gives the properties of the $Z$- and $Y$-operators at each level $i$ of the induction argument. We take here as convention that in the indexing of the spaces and operators, the index $0$ refers to the base case developed in sections \ref{QSecFormProb} and \ref{QSecFieldRecursion}. Also, we denote by $\BL_{11}^{(i)}$ the restriction of the operator $\BGG_1^{(i)} \BL \BGG_1^{(i)}$ on $\CE^{(i)}$.
\begin{Thm}
Assume that the operators $\BL$ on $\CH^{(0)}$ and $\BL_{11}^{(0)}$ on $\CE^{(0)}$ are invertible and let $i\in\mathbb{N}$. If, for each $k=1,....,i$, we have $\operatorname{Ker}\BL_{11}^{(k)}=\{0\}$ then for all $k=0,....,i$, the $Z$-operator $\BLE^{(k)}: \CU^{(k)} \to \CU^{(k)}$ is well-defined by both representation formulas:
\begin{eqnarray}
\BLE^{(k)}&=&\BGG_0^{(k)} \BL \BGG_0^{(k)}- \BGG_0^{(k)}\BL\BGG_1^{(k)}(\BGG_1^{(k)}\BL \BGG_1^{(k)})^{-1} \BGG_1^{(k)} \BL \BGG_0^{(k)} ] \label{Qeq.schurLEK} \\
&=& \BGG_0^{(k)}[(\BGG_0^{(k)}+\BGG_2^{(k)})\BL^{-1}(\BGG_0^{(k)}+\BGG_2^{(k)})]^{-1}\BGG_0^{(k)}, \label{Qeq.LEKrep2}
\end{eqnarray}
and $\BLE^{(k)}$ satisfies the homogeneity, normalization, and Herglotz properties. Furthermore, if $\CV^{(k)}\cap \CJ^{(k)}=\{{\bf 0}\}$ for each $k=0,....,i$, then the $Y$-operator $\BYE^{(k)}:\CV^{(k)}\to \CV^{(k)}$ is also well-defined by
\begin{align}\label{Qeq.defBYEk}
\BYE^{(k)}=\BGP_1^{(k)}\BGG_2^{(k)}(\BGG_2^{(k)}\BGP_2^{(k)}{\BL}^{-1}\BGP_2^{(k)}\BGG_2^{(k)})^{-1}\BGG_2^{(k)}\BGP_1^{(k)},
\end{align}
and $\BYE^{(k)}$ satisfies both the homogeneity and Herglotz properties. Moreover, if we also suppose that $$\CV^{(k)}\cap\CE^{(k)}=0 \mbox{ and }  \CP_1^{(k)}\cap\CU^{(k)}=\CP_2^{(k)}\cap \CU^{(k)}=\{{\bf 0}\},\ \mbox{ for all } k=0,...,i-1$$
then for all $k=0,...,i-1$ we have
\begin{align}
\dim \CU^{(0)}=\dim \CU^{(k)}=\dim \CV^{(k)}=\dim \CU^{(i)},
\end{align}
and there exists an invertible operator $\BK^{(k)}:\CU^{(k+1)}\to \CV^{(k)}$ defined by
\begin{align}\label{Qeq.recKi2}
\BK^{(k)}=-(\BGP_1^{(k)}\BGG_1^{(k)}\BGP_1^{(k)})^{-1}\BGP_1^{(k)}\BGG_1^{(k)}\BGP_2^{(k)},
\end{align}
such that $\BYE^{(k)}$ and $\BLE^{(k+1)}$ are linked by the following relation:
\begin{align}\label{Qeq.recYi2}
\BYE^{(k)}=\BK^{(k)}\BLE^{(k+1)}\, (\BK^{(k)})^{\dag}.
\end{align}
\end{Thm}
\begin{proof}
Suppose the operators $\BL$ on $\CH^{(0)}$ and $\BL_{11}^{(0)}$ on $\CE^{(0)}$ are invertible and fix an integer $i\in\mathbb{N}$. First, by Proposition \ref{Qprop.equivrepresLE}, in the case $k=0$ the $Z$-operator $\BLE^{(0)}:\CU^{(0)}\to \CU^{(0)}$ is well-defined by both formulas (\ref{Qeq.schurLEK}) and (\ref{Qeq.LEKrep2}) and hence satisfies the homogeneity, normalization, and Herglotz properties.%
\index{homogeneity property}%
\index{normalization property}%
\index{Herglotz property}

Next, suppose that for each $k=1,....,i$, we have $\operatorname{Ker}\BL_{11}^{(k)}=\{0\}$. Hence, as $\BL_{11}^{(0)}$ is invertible, we deduce by induction by using the Proposition \ref{QProp.defLEschur} that all the operators $\BL_{11}^{(1)},\ldots, \BL_{11}^{(i)}$ are invertible and that the $Z$-operator $\BLE^{(k)}:\CU^{(k)}\to \CU^{(k)}$ is well-defined by formula (\ref{Qeq.schurLEK}). %Then it follows from Propositions \ref{QPropFredholmBlockEquiv} and \ref{QPropFinDimSpaceSep} that for all $k=1,\ldots,i$, the operator $\BL_{11}^{(k)}$ is invertible on $\CE^{(k)}$ %
Then, as $\BL$ is assumed invertible on $\CH$,  we have $L_1\neq 0$  and $L_2\neq 0$ (see Remark \ref{QRem.inv}) and this implies that $\BL$ is  invertible on  $\CH^{(k)}$ too (see Remark \ref{QRemLOnSubspaceRestr}). Hence, it follows from Proposition \ref{Qprop.equivrepresLE} that the $Z$-operator $\BLE^{(k)}$ is also well-defined by the formula (\ref{Qeq.LEKrep2}) and thus satisfies the homogeneity, normalization, and Herglotz properties.

Suppose now that in addition, $\CV^{(k)}\cap \CJ^{(k)}=\{{\bf 0}\}$ for each $k=0,....,i$. Then it follows immediately from Proposition \ref{Qprop.defYE} that the $Y$-operator $\BYE^{(k)}:\CV^{(k)}\to \CV^{(k)}$ is also well-defined by the formula (\ref{Qeq.defBYEk}) and hence satisfies both the homogeneity and Herglotz properties for all $k=0,....,i$.

Finally, suppose that in addition, $$\CV^{(k)}\cap\CE^{(k)}=0 \mbox{ and }  \CP_1^{(k)}\cap\CU^{(k)}=\CP_2^{(k)}\cap \CU^{(k)}=\{{\bf 0}\},\ \mbox{ for all } k=0,...,i-1.$$ Then the rest of the proof of this theorem now follows immediately from Theorem \ref{Qthm.relYStarAndLStar1}. This completes the proof.
\end{proof}
%\section{The effective conductivity tensor}
%\setcounter{equation}{0}
%%%%%%%%%%%%%%%%%%%%%%%%%%%%%%%%%%%%%%%%%%%%%%%%%%%%%%%%%%%%%%%%%%%

\section*{Acknowledgements}
Graeme Milton is grateful to the National Science Foundation for support through the Research Grant DMS-1211359. Aaron Welters is grateful for the support from the U.S.\ Air Force Office of Scientific Research (AFOSR) through the Air Force's Young Investigator Research Program (YIP) under the grant FA9550-15-1-0086.

\bibliographystyle{mod-xchicago}
%\bibliographystyle{siam}
%\bibliographystyle{/u/ma/milton/tex/mod-xchicago}
%\bibliography{/u/ma/milton/tcbook,/u/ma/milton/newref}
\bibliography{newref,tcbook}

\end{document}